\documentclass[aoas]{imsart}
\RequirePackage{amsthm,amsmath,amssymb,bm}
\RequirePackage{graphicx}
\RequirePackage{epstopdf}
\RequirePackage{natbib}
\RequirePackage{rotating}
\RequirePackage[colorlinks,citecolor=blue,urlcolor=blue]{hyperref}

\startlocaldefs
\newcommand{\iid}{\stackrel{\mathrm{iid}}{\sim}}
\newcommand{\reals}{\mathbb{R}}

\newtheorem{lemma}{Lemma}
\newtheorem{proposition}{Proposition}
\newtheorem{corollary}{Corollary}

\endlocaldefs

\begin{document}

\begin{frontmatter}

\title{Full Matching Approach to Instrumental Variables Estimation with Application to the Effect of Malaria on Stunting}
\runtitle{Full Matching IV Estimation}


\author{\fnms{Hyunseung} \snm{Kang}\thanksref{m1}\corref{}\ead[label=e1]{khyuns@wharton.upenn.edu}}
\author{\fnms{Benno} \snm{Kreuels}\thanksref{m2,m3}\ead[label=e2]{kreuels@bni-hamburg.de}}
\author{\fnms{J\"{u}rgen} \snm{May}\thanksref{m3}\ead[label=e3]{may@bni.uni-hamburg.de}}
\and
\author{\fnms{Dylan} S. \snm{Small}\thanksref{m1}\ead[label=e4]{dsmall@wharton.upenn.edu}}

\runauthor{H. Kang et al.}

\affiliation{University of Pennsylvania\thanksmark{m1}, University Medical Centre\thanksmark{m2} and
                   Bernhard Nocht Institute for Tropical Medicine\thanksmark{m3}}

\address{Department of Statistics \\
The Wharton School \\
University of Pennsylvania \\
Philadelphia, PA, USA19104 \\
\printead{e1} \\
\phantom{E-mail:\ }\printead*{e4}}

\address{Division of Tropical Medicine \\
I. Department of Internal Medicine \\
University Medical Centre \\
Hamburg, Eppendorf, Germany \\
\printead{e2}}

\address{Research Group Infectious Disease Epidemiology \\
Bernhard Nocht Institute for Tropical Medicine \\
Hamburg, Germany \\
\printead{e2} \\
\phantom{E-mail:\ } \printead*{e3}}

\begin{abstract}
Most previous studies of the causal relationship between malaria and stunting have been studies where potential confounders are controlled via regression-based methods, but these studies may have been biased by unobserved confounders.  Instrumental variables (IV) regression offers a way to control for unmeasured confounders where, in our case, the sickle cell trait can be used as an instrument. However, for the instrument to be valid, it may still be important to account for measured confounders.  The most commonly used instrumental variable regression method, two-stage least squares, relies on parametric assumptions on the effects of measured confounders to account for them.  Additionally, two-stage least squares lacks transparency with respect to covariate balance and weighing of subjects and does not blind the researcher to the outcome data. To address these drawbacks, we propose an alternative method for IV estimation based on full matching. We evaluate our new procedure on simulated data and real data concerning the causal effect of malaria on stunting among children.  We estimate that the risk of stunting among children with the sickle cell trait decreases by 0.22 per every malaria episode prevented by the sickle cell trait, a substantial effect of malaria on stunting (p-value: 0.011, 95\% CI: 0.044, 1).
\end{abstract}



\end{frontmatter}

\section{Introduction}

\subsection{Motivation: Does malaria cause stunting?} \label{sec:intro.motivatingExample}
From January 2003 to January 2004, 1070 infants from Ghana, Africa were recruited to a clinical trial on Intermittent Preventative Treatment for malaria (IPT) \citep{kobbe_randomized_2007}. From the time of recruitment at 3 months of age until two years of age, each child was monitored monthly for the presence of malaria parasites with measurements every three months of length/height. Table \ref{tab:char} lists the baseline characteristics of the 1070 infants in our data.

One of the public health questions of interest from this clinical study was whether malaria caused stunted growth among children. In 2013 alone, there were 128 million estimated cases of malaria in sub-Saharan Africa, with most cases occurring in children under the age of 5 \citep{who_world_2014}. Stunting, defined as a child's height being two standard deviations below the mean for his/her age, is a key indicator of child development \citep{who_who_2006}. If malaria does cause stunted growth, several intervention strategies can be implemented to mitigate stunted growth, such as distribution of mosquito nets, control of the mosquito population during seasons of high malarial incidence, and surveillance of mosquito populations. 

The current body of evidence suggests that there is a strong positive relationship between malaria exposure and stunted growth \citep{genton_relation_1998, deen_increased_2002,nyakeriga_malaria_2004, ehrhardt_malaria_2006,fillol_impact_2009, deribew_malaria_2010,crookston_exploring_2010}. Unfortunately, a fundamental limitation with these prior studies is that they are observational studies and consequently, there is always a concern that important confounders were not controlled for. For example, \citet{fillol_impact_2009} and \citet{deribew_malaria_2010} stated that a limitation in their studies was not controlling for diet, specifically a child's intake of micronutrients such as vitamins, zinc, or iron as these micronutrients could impact a child's growth as well as his immune system's ability to fight off a malaria episode. In addition, \citet{ehrhardt_malaria_2006} and \citet{crookston_exploring_2010} suggested controlling for socioeconomic factors in future studies of malaria and malnutrition because affluent families are more likely to provide mosquito nets and nutritious food to their children compared to impoverished families. Short of a randomized clinical trial, which is unethical in this context, unmeasured confounders are likely present in all the aforementioned studies, because of the practical limitations of accounting for all possible confounders.


\subsection{Instrumental variables and sickle cell trait} \label{sec:intro.IV}
Instrumental variables (IVs) is an alternative method to estimate the causal effect of an exposure on the outcome when there is unmeasured confounding, provided that a valid instrument is available \citep*{angrist_identification_1996,hernan_instruments_2006, brookhart_preference-based_2007, cheng_semiparametric_2009, swanson_commentary_2013, baiocchi_instrumental_2014}. The core assumptions for a variable to be a valid instrumental variable are that the variable (A1) is associated with the exposure, (A2) has no direct pathways to the outcome, and (A3) is not associated with any unmeasured confounders after controlling for the measured confounders (See Figure \ref{fig:1} and Section \ref{sec:defIV} for more detailed discussions). If measured covariates are available, like in our data, the plausibility of the instrument satisfying the three core assumptions can be improved by conditioning on the covariates, especially (A3).
\begin{figure}
\vspace{6pc}
\includegraphics[scale=0.2]{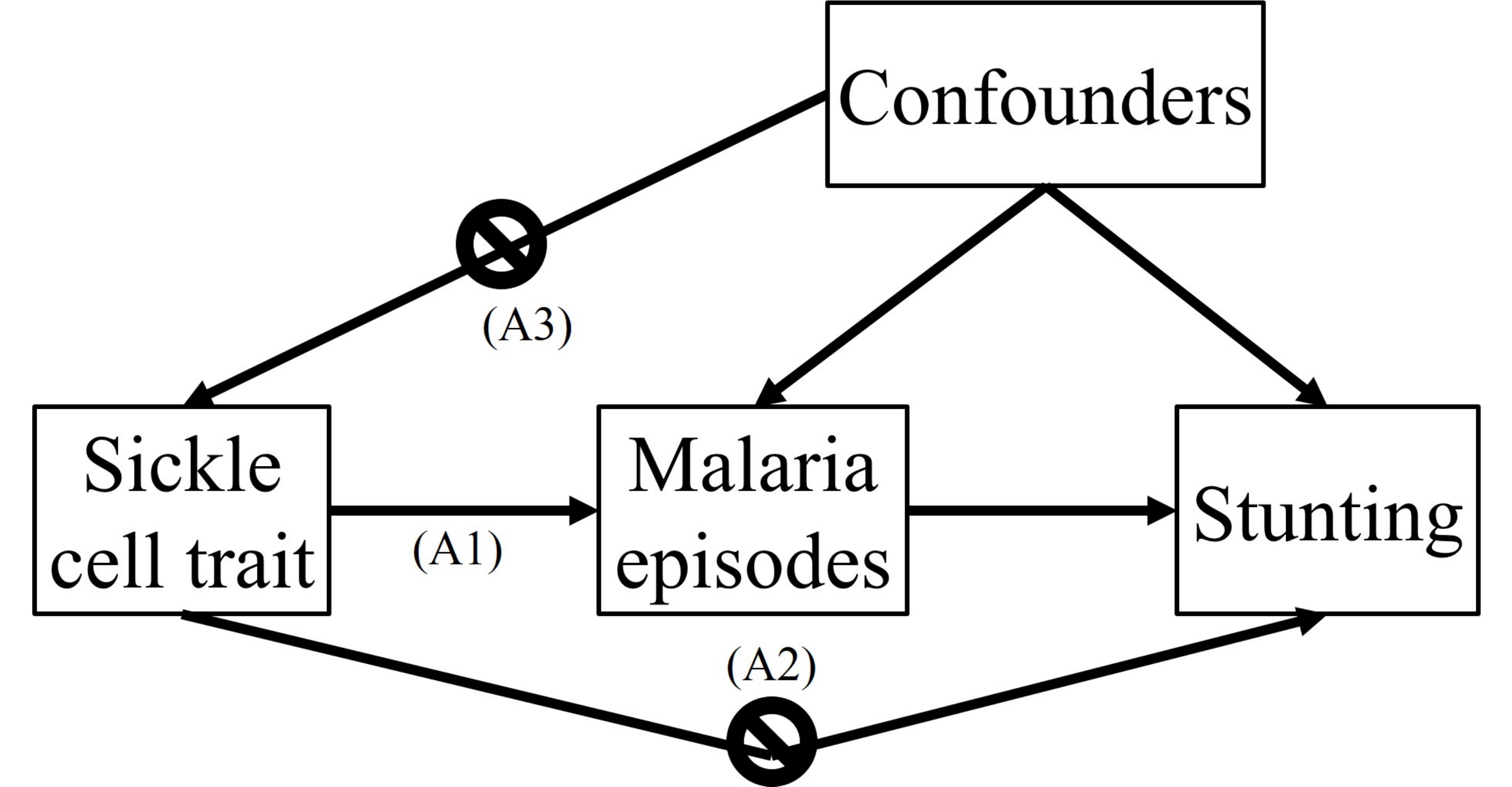}
\caption[]{Causal diagram for the malaria study. Numbers (A1,A2,A3) represent MR assumptions.}
\label{fig:1}
\end{figure} 
For our study of analyzing the causal effect of malaria on stunting, we follow a recent approach by \citet{davey_smith_mendelian_2003} and especially \citet{kang_causal_2013} where genotypic variations are used as instruments and propose to use the presence of a sickle cell genotype (HbAS) versus carrying the normal hemoglobin type (HbAA) as an instrument. The sickle cell genotype (HbAS) is a condition where a person inherits from one parent a mutated copy of the hemoglobin beta (HBB) gene called the sickle cell gene mutation that bends red blood cells into a sickle (crescent) shape, but inherits a normal copy of the HBB gene from the other parent. The sickle cell trait protects against malaria, but is thought to be otherwise mostly asymptomatic \citep{aidoo_protective_2002}. Note that we exclude from the analysis people who have two copies of the HBB gene, i.e. people who suffer from sickle cell disease which causes severe symptoms; sickle cell disease (two copies of the HBB gene) is thought to persist despite its evolutionary disadvantage because of the sickle cell trait (one copy of the HBB gene) protecting against malaria \citep{may_hemoglobin_2007}. We discuss in detail the validity of the sickle cell trait IV in Section \ref{sec:defIV}. In addition, we propose to combine the covariates that were already measured for this data in Table \ref{tab:char} to increase the plausibility of our sickle cell trait being a valid instrument.

\subsection{Two-stage least squares} \label{sec:intro.tsls}
The most popular and well-studied among methods that use an IV and measured covariates to estimate causal effects is two-stage least squares (2SLS) \citep{angrist_does_1991, card_using_1995, wooldridge_econometrics_2010}. For example, in \citet{card_using_1995}, which studied the effect of education on wages, 2SLS with proximity to a 4-year college as an IV was used to control for measured covariates such as race and parents' education. Specifically, 2SLS first estimated, via least squares, the predicted exposure (education) given the instrument (proximity to 4-year college) and the measured covariates, and second, regressed the outcome (earnings) on this predicted exposure and the measured covariates; the 2SLS estimate of the causal effect is the coefficient on the predicted exposure in the second regression. Standard results in econometrics show 2SLS estimators are consistent and efficient under linear single-variable structural equation models with a constant treatment effect \citep{wooldridge_econometrics_2010}. When treatment effects are not constant, \citet{angrist_two-stage_1995} showed that under certain monotonicity assumptions, 2SLS converges to a weighted average of the covariate-specific treatment effects with the weights proportional to the average conditional variance of the expected value of the treatment given the covariates and the instrument. Other IV methods to estimate causal effects in the presence of measured covariates include Bayesian methods \citep{imbens_bayesian_1997}, semiparametric methods \citep{abadie_semiparametric_2003, tan_regression_2006, ogburn_doubly_2015}, and nonparametric methods \citep{frolich_nonparametric_2007}.

Despite its attractive estimation properties, 2SLS has some drawbacks in (i) lack of transparency of the population to which the estimate applies, (ii) lack of blinding of the analyst/researcher and (iii) dependence on parametric assumptions. First, with regards to transparency, suppose that there are some values of the covariates for which the instrument is almost always low, some values for which the instrument is almost always high and some values of the covariates for which the instrument takes on both low and high values.  Then, the 2SLS estimate will put most of its weight on the causal effect for subjects with the values of the covariates for which the instrument takes on both low and high values, and little weight on subjects with the values of the covariates for which the instrument usually takes on low (or high) values.  In our malaria study, this would mean that there might be some villages (a covariate) that are receiving little weight in the 2SLS estimate; consequently, the 2SLS estimate might not be helpful for understanding the effect of malaria on stunting in these villages even though these villages might have contributed many subjects to the analysis.  Although the weighting function in 2SLS can be studied, there is nothing in the 2SLS estimation procedure itself that warns us when some values of the covariates are receiving little weight and it is rare to see discussion of the weighting function for 2SLS in empirical papers.  

Second, 2SLS lacks blinding with respect to the outcome data when adjusting for covariates. \citet{cochran_planning_1965}, \citet{rubin_design_2007} and \citet{rosenbaum_design_2010} argue that the best observational studies resemble randomized experiments. An important feature of the design of randomized experiments is that when designing the study and planning the analysis, the researcher is blinded to the outcome data. However, in regression-based procedures for adjusting for covariates like 2SLS, there is often judgment that needs to be exercised in choosing covariate adjustment models, which require one to look at the outcome data and estimates of causal effects to exercise such judgment. It is difficult even for the most honest researcher to be completely objective in comparing models when the researcher has an a priori hypothesis or expectation about the direction of the causal effect \citep{rubin_estimating_2006}. 

Third, 2SLS relies on proper specification of how the measured covariates affect the outcomes. Often, parametric modeling assumptions are made for how the measured confounders affect the outcome. In particular, 2SLS, as usually implemented, relies on the measured confounders having a linear effect on the expected outcome. Section \ref{sec:simulation.2sls} contains simulation evidence about 2SLS that demonstrates its reliance on linear, parametric assumptions. 

\subsection{Instrumental variables with full matching} \label{sec:intro.fullmatch}
Matching is an alternative method to adjust for measured covariates. A matching algorithm groups individuals in the data with different values of the instrument but similar values of the observed covariates, so that within each group, the only difference between the individuals is their values of the instrument \citep*{haviland_combining_2007, rosenbaum_design_2010, stuart_matching_2010}. For example, in the malaria data, a matching algorithm seeks to produce matched sets so that in a matched set, individuals are born in the same village and are similar on other measured covariates. The only difference between individuals in a matched set is their instrument values. We can then compare stunting between individuals with high and low values of the instrument within a matched set to assess the effect of malaria on stunting \citep{baiocchi_building_2010}.

Matching addresses the drawbacks of 2SLS discussed in the previous section. First, if there are values of covariates for which almost all subjects have a high (or low) value of the IV, then the matching algorithm and associated diagnostics will tell us that matched sets cannot be formed when subjects in the matched sets have certain values of the covariates but different levels of the IV; thus, it will be transparent that for these values of the covariates, the causal effect cannot be estimated without extrapolation.  Relatedly, matching allows us to control the weighting of subjects with different values of the covariates to make the weighting transparent, such as weighting the covariates in proportion to their population frequency. Second, matching is blind to the outcome data; a matching algorithm only requires the measured covariates and the instrument values for each individual in the data.  Diagnostics can be done and the matching can be adjusted until it is adequate, all without looking at the outcome data. Finally, when estimating the causal effect, matching makes non-parametric inference; it does not use any parametric assumptions model such as linearity and parametric assumptions on the model. 
 
Previous work using matching in studying causality is abundant in non-IV settings; see \citet{stuart_matching_2010} for a complete overview. In contrast, work on using matching methods on IV estimation is limited to pair matching \citep{baiocchi_building_2010} and fixed control matching, i.e. each unit with level 1 of the IV is matched to a fixed number of units with level 0 of the IV \citep{kang_causal_2013}. A drawback to these matching methods is that they do not use the full data \citep{keele_stronger_2013,zubizarreta_stronger_2013}. In particular, \citet{kang_causal_2013} studied the same causal effect of interest, malaria on stunting, but with a smaller amount of data, because the statistical methodology was limited to matching with fixed controls. That is, out of the total of 884 individuals available, the matching algorithm dropped 25\% of the individuals and the final statistical inference was based only on 660 individuals. 

In this paper, we develop an IV full matching approach that uses the full data. Full matching is the most general, flexible, and optimal type of matching \citep{rosenbaum_characterization_1991, hansen_full_2004, rosenbaum_design_2010}. Specifically, full matching is the generalization of any type of matching, such as pair matching, matching with fixed controls, or matching with variable controls. Full matching is also flexible in that it can incorporate constraints on matched set structures, such as limiting the number of individuals in each matched set, to improve statistical efficiency.  Finally, full matching is optimal in the sense that it produces matched sets where within each set, measured covariates between individuals with different instrument values are most similar \citep{rosenbaum_characterization_1991}.

Under IV estimation with full matching, we derive a randomization-based testing procedure and sensitivity analysis based on the proposed test statistic. We conduct simulation studies to study the performance of 2SLS versus full matching IV estimation, specifically analyzing the robustness of both methods to non-linearity (Section \ref{sec:simulation.2sls}). In the same spirit, we also conduct simulation studies to compare our full matching IV estimation with another nonparametric method introduced by \citep{frolich_nonparametric_2007} introduced in Section \ref{sec:intro.tsls}, specifically looking at bias and variance between the two nonparametric methods. Finally, we apply full matching IV estimation to analyze the causal effect of malaria on stunting and demonstrate the full matching method's transparency in adjusting for covariates. 

\section{Methods}
\subsection{Notation} \label{sec:notation}
To introduce the idea of matching in IV estimation, we introduce the following notation. Let $i = 1,\ldots,I$ index the $I$ total matched sets that individuals are matched into. Each matched set $i$ contains $n_i \geq 2$ subjects who are indexed by $j = 1,\ldots,n_i$ and there are a total of $N = \sum_{i=1}^I n_i$ individuals in the data. Let $Z_{ij}$ denote a binary instrument for subject $j$ in matched set $i$. In each matched set $i$, there are $m_i$ subjects with $Z_{ij} = 1$ and $n_i - m_i$ subjects with $Z_{ij}=0$. For instance, in the malaria data, for each $i$th matched set, there are $m_i$ children who inherited the sickle cell trait, HbAS (i.e. $Z_{ij} =1$), and $n_i - m_i$ children who inherited HbAA (i.e. $Z_{ij} = 0$). Let $\mathbf{Z}$ be a random variable that consists of the collection of $Z_{ij}$'s, $\mathbf{Z} = (Z_{11},Z_{12},....,Z_{I,n_I})$. Define $\Omega$ as the set that contains all possible values $\mathbf{z}$ of $\mathbf{Z}$, so $\mathbf{z} \in \Omega$ if $z_{ij}$ is  binary and $\sum_{j=1}^{n_i} z_{ij}=m_i$ for all $I$ matched sets. Thus, the cardinality of $\Omega$, denoted as $|\Omega|$, is $|\Omega| = \prod_{i=1}^I {n_i \choose m_i} $. Denote $\mathcal{Z}$ to be the event that $\mathbf{Z} \in \Omega$. 

For individual $j$ in matched set $i$, define $d_{1ij}$ and $d_{0ij}$ to be the potential exposure values under $Z_{ij} = 1$ or $Z_{ij} = 0$, respectively. With the malaria data, $d_{1ij}$ and $d_{0ij}$ represent the number of malaria episodes the child would have if she had the sickle cell trait, $Z_{ij} =1$, and no sickle cell trait, $Z_{ij} = 0$, respectively. Also, define $r_{1ij}^{(k)}$ to be the outcome individual $i$ would have if she were assigned instrument value $1$ and level $k$  of the exposure, and $r_{0ij}^{(k)}$ to be the outcome individual $i$ would have if she were assigned instrumental value $0$ and level $k$ of the exposure. Then, $r_{1ij}^{(d_{1ij})}$ and $r_{0ij}^{(d_{0ij})}$ are the potential outcomes if the individual were assigned levels 1 and 0 of the instrument respectively and the exposure took its natural level given the instrument, resulting in exposures, $d_{1ij}$ and $d_{0ij}$, respectively. In the malaria data, $r_{1ij}^{(d_{1ij})}$ is a binary variable that represents whether the $j$th child in the $i$th matched set would be stunted (i.e. 1) or not (i.e. 0) if the child carried the sickle cell trait (i.e. if $Z_{ij} =1$) and $r_{0ij}^{(d_{0ij})}$ is a binary variable that represents whether the child would be stunted or not if the child carried no sickle cell trait (i.e. if $Z_{ij} = 0$). The potential outcome notations assume the Stable Unit Treatment Value Assumption where (i) an individual's outcome and exposure depend only on her own value of the instrument and not on other people's instrument values and (ii) an individual's outcome only depends on her own value of the exposure and not on other people's exposure  \citep{rubin_comment_1980}.

For individual $j$ in matched set $i$, let $R_{ij}$ be the binary observed outcome and $D_{ij}$ be the observed exposure. The potential outcomes $r_{1ij}^{(d_{1ij})},r_{0ij}^{(d_{0ij})},d_{1ij}$, and $d_{0ij}$ and the observed values $R_{ij}, D_{ij}$, and  $Z_{ij}$ are related by the following equation:
\begin{equation} \label{eq:potObs}
R_{ij} = r_{1ij}^{(d_{1ij})} Z_{ij} + r_{0ij}^{(d_{0ij})}(1 - Z_{ij}) \qquad{} D_{ij} = d_{1ij} Z_{ij} + d_{0ij} (1 - Z_{ij})
\end{equation}
For individual $j$ in matched set $i$, let $\mathbf{X}_{ij}$ be a vector of observed covariates and $u_{ij}$ be the unobserved covariates. For example, in the malaria data, $\mathbf{X}_{ij}$ represents each child's covariates listed in Table \ref{tab:char} while $u_{ij}$ is an unmeasured confounder, like diet, which was mentioned in Section \ref{sec:intro.motivatingExample}. We define the set $\mathcal{F} = \{(r_{1ij}^{(d_{1ij})},r_{0ij}^{(d_{0ij})},d_{1ij},d_{0ij},\mathbf{X}_{ij},u_{ij}),i=1,...,I,j=1,...,n_i\}$ to be the collection of potential outcomes and all covariates/confounders, observed and unobserved.

\subsection{Full matching algorithm} \label{sec:fullmatch}
A matching algorithm controls the bias resulting from different observed covariates by creating $I$ matched sets indexed by $i$, $i=1,\ldots, I$ such that individuals within each matched set have similar covariate values $\mathbf{x}_{ij}$ and the only difference between individuals in each matched set is their instrument values, $Z_{ij}$. In a full matching algorithm, each matched set $i$ either contains $m_i =1$ individual with $Z_{ij} = 1$ and $n_i - 1$ individuals with $Z_{ij} = 0$ or $m_i = n_i -1$ individuals with $Z_{ij} = 1$ and 1 individual with $Z_{ij} = 0$.

\citet{rosenbaum_observational_2002, rosenbaum_design_2010}, \citet{hansen_full_2004}, and \citet{stuart_matching_2010} provide an overview of matching and a discussion on various distance metrics and tools to measure similarity for observed and missing covariates. For the malaria data, Section \ref{sec:data.match} describes how we used propensity score caliper matching with rank-based Mahalanobis distance to measure covariate similarity. Once we have obtained the distance matrix, we use an R package available on CRAN called \emph{optmatch} developed by \citet{hansen_optimal_2006} to find the optimal full matching.

\subsection{Conditions for sickle cell trait as a valid instrument} \label{sec:defIV}
We formalize the core assumptions of an instrumental variable below \citep{holland_causal_1988, angrist_identification_1996,yang_dissonant_2014} (see Figure \ref{fig:1}).
\begin{enumerate}
\item[(A1)] The instrument must be associated with the exposure, or in $\mathcal{F}$, $\sum_{i=1}^I$ $\sum_{j=1}^{n_i} (d_{1ij} - d_{0ij}) \neq 0$
\item[(A2)] The instrument can only affect the outcome if it affects the exposure. Since the $r$'s don't depend on $z$ under this assumption, we write $r_{1ij}^{(k)} = r_{0ij}^{(k)} \equiv r_{ij}^{(k)}$ for all $k$ in $\mathcal{F}$ (exclusion restriction)
\item[(A3)] The instrument is effectively randomly assigned within a matched set, $P(Z_{ij} = 1 | \mathcal{F}, \mathcal{Z}) = m_i/n_i$ for each $i$.
\end{enumerate}
In Figure \ref{fig:1}, (A1) corresponds to there being an association between the instrument and the exposure, (A2) corresponds to that all directed pathways from the instrument to the outcome pass through the exposure and (A3) corresponds to the instrument, conditional on measured variables, being unassociated with unmeasured variables that are associated with the outcome.

We now assess the validity of (A1)-(A3) for the sickle cell trait, the instrument for our analysis on the effect of malaria on stunting. For assumption (A1), there is substantial evidence that the sickle cell trait does provide protection against malaria as compared to people with two normal copies of the HBB gene (HbAA) \citep{aidoo_protective_2002, williams_immune_2005, may_hemoglobin_2007, cholera_impaired_2008, kreuels_differing_2010}. For assumption (A2), this could be violated if the sickle cell trait had effects on stunting other than through causing malaria, for instance, if the sickle cell trait was pleiotropic  \citep{davey_smith_mendelian_2003}. We can partially test this assumption by examining individuals who carry the sickle cell trait, but who grew up in a region where malaria is not present. That is, if assumption (A2) were violated, heights between individuals with HbAS and HbAA in such a region would be different since there would be a direct arrow between the sickle cell trait and height. We examined studies among African American children and children from the Dominican Republic and Jamaica for whom the sickle cell trait is common, but there is no malaria in the area. These two regions also match nutritional and socioeconomic conditions that are closer to our study population in Ghana so that the populations (and subsequent subpopulations among them) are comparable. From these studies from the regions, we found no evidence that the sickle cell trait affected a child's physical development \citep{ashcroft_growth_1976,kramer_growth_1978, ashcroft_heights_1978, rehan_growth_1981}. This supports the validity of assumption (A2). 

Although the results of this test support the validity of (A2), (A2) could still be violated. For example, the regions we use to support assumption (A2) may be different than Ghana through unmeasured characteristics, which would make the populations incomparable. As another example, the sickle cell trait could have a direct effect that interacts with the environment in such a way that the direct effect is only present in Africa, but not in the United States, the Dominican Republic, or Jamaica. One specific point of concern raised by a referee is iron supplements. In the malaria study that we are considering, children with low hemoglobin received iron supplements and iron supplements can reduce the risk of stunting. A potential concern is that the sickle cell trait may induce a child to have low iron levels, thereby increasing the risk of stunting without going through the malaria pathway in Figure \ref{fig:1} and violating (A2). However, \cite{kreuels_differing_2010} found that in the malaria study we are considering, children carrying HbAA tend to have lower hemoglobin levels than children carrying HbAS. Thus, children with the sickle cell trait,  HbAS, were less likely to receive iron supplements. Consequently, if there's a violation of the exclusion restriction because of iron supplements, it would tend to bias our estimate of the increase in stunting from malaria downwards and our estimate can be regarded as a conservative estimate of the effect of malaria on increasing stunting.

For assumption (A3), this assumption would be questionable in our data if we did not control for any population stratification covariates. Population stratification is a condition where there are subpopulations, some of which are more likely to have the sickle cell trait, and some of which are more likely to be stunted through mechanisms other than malaria \citep{davey_smith_mendelian_2003}. For example, in Table \ref{tab:char} which provides the baseline characteristics for our data, we observed that the village Tano-Odumasi had more children with HbAA than HbAS. It is possible that there are other variables besides HbAA that differ between the village Tano-Odumasi and other villages and affect stunting. Hence, assumption (A3) is more plausible if we control for observed variables like village of birth. Specifically, within the framework of full matching, for each matched set $i$, if the observed variables $\mathbf{x}_{ij}$ are similar among all $n_i$ individuals, it may be more plausible that the unobserved variable $u_{ij}$ plays no role in the distribution of $Z_{ij}$ among the $n_i$ children. If (A3) exactly holds and subjects are exactly matched for $X_{ij}$, then within each matched set $i$, $Z_{ij}$ is simply a result of random assignment where $Z_{ij} = 1$ with probability $m_i / n_i$ and $Z_{ij} = 0$ with probability $(n_i - m_i)/ n_i$ when we condition on the number of units int he matched set with $Z_{ij} = 1$ being $m_i$. In Section \ref{sec:sens}, we discuss a sensitivity analysis that allows for the possibility that even after matching for observed variables, the unobserved variable $u_{ij}$ may still influence the assignment of $Z_{ij}$ in each matched set $i$, meaning that assumption (A3) is violated.

There are also other assumptions associated with instrumental variables, most notably the Stable Unit Treatment Value Assumption (SUTVA) in Section \ref{sec:notation} and the monotonicity assumption in \citet{angrist_identification_1996}. SUTVA, within the framework of MR, states that one individual's potential outcomes are not affected by the exposures and genotype assignments of other individuals given the individual's exposure and genotype assignment, and one individual's potential exposure is not affected by the genotype assignment of other individuals given the individual's own genotype assignment \citep{angrist_identification_1996}. This is fairly reasonable in our setting. The outcome, stunting, given an individual's own malaria exposure and HbAS status, should not be affected by others'  malaria exposure and HbAS. The exposure would be affected by others' HbAs status if HbAS affected malaria transmission. However, there is no evidence that HbAs protects against parasitemia and hence, there is no evidence that HbAS affects transmission; HbAS's effect appears to be limited to protection against severe disease manifestations from malaria (90\%) and mild disease manifestations (30\%) \citep{kreuels_differing_2010, taylor_haemoglobinopathies_2012}. 

Monotonicity, within the framework of MR, states that there are no individuals who would experience an adverse effect on the exposure from inheriting the genotype which is purported to bring positive effect on the exposure. In our study, monotonicity is plausible because there are known biological mechanisms by which the sickle cell genotype protects against malaria \citep{friedman_erythrocytic_1978, friedman_biochemistry_1981, williams_immune_2005, cholera_impaired_2008} and no known mechanisms by which the sickle cell genotype increase the risk of malaria.

\subsection{Effect ratio} \label{sec:effectRatio}
We define the parameter of interest, called the \emph{effect ratio}, which is a parameter of the finite population of $N = \sum_{i=1}^I n_i$ individuals characterized by $\mathcal{F}$.
\begin{equation} \label{eq:effectRatio}
\lambda = \frac{ \sum_{i=1}^I  \sum_{j=1}^{n_i} r_{1ij}^{(d_{1ij})} - r_{0ij}^{(d_{0ij})}}{\sum_{i=1}^I \sum_{j=1}^{n_i} d_{1ij} - d_{0ij}}
\end{equation}
The effect ratio is the change in the outcome caused by the instrument divided by the change in the exposure caused by the instrument. The effect ratio can be identified by taking the ratio of the differences in expected values.
\begin{equation} \label{eq:effectRatioID}
\lambda = \frac{\sum_{i=1}^I \sum_{j=1}^{n_i} E(R_{ij} | Z_{ij} = 1, \mathcal{F}, \mathcal{Z}) - E(R_{ij} | Z_{ij} = 0, \mathcal{F}, \mathcal{Z})}{\sum_{i=1}^I \sum_{j=1}^{n_i} E(D_{ij} | Z_{ij} = 1, \mathcal{F}, \mathcal{Z}) - E(D_{ij} | Z_{ij} = 0, \mathcal{F}, \mathcal{Z})}
\end{equation}
The effect ratio also admits a well-known interpretation in IV literature if all the IV assumptions, (A1)-(A3), and the monotonicity assumption whereby $d_{1ij} \geq d_{0ij}$ for every $i,j$ in $\mathcal{F}$, are satisfied. Specifically, suppose $d_{1ij}$ and $d_{0ij}$ are discrete values from $0$ to $M$, which is the case with the malaria data where $d_{1ij}$ and $d_{0ij}$ are the number of malaria episodes. Then, in Proposition 1 of the supplementary article \citep{kang_supp_2015}, we show that
\begin{equation} \label{eq:effectRatioCACE}
\lambda =  \sum_{i=1}^I \sum_{j=1}^{n_i} \sum_{k=1}^M (r_{ij}^{(k)} - r_{ij}^{(k-1)}) w_{ijk} 
\end{equation}
where 
\[
w_{ijk} = \frac{\chi(d_{1ij} \geq k > d_{0ij})}{\sum_{i=1}^I \sum_{j=1}^{n_i} \sum_{l=1}^M \chi(d_{1ij} \geq l > d_{0ij})}
\]
and $\chi(\cdot)$ is an indicator function. In words, with the IV assumptions and the monotonicity assumption, the effect ratio is interpreted as the weighted average of the causal effect of a one unit change in the exposure among individuals in the study population whose exposure would be affected by a change in the instrument. Each weight $w_{ijk}$ represents whether an individual $j$ in stratum $i$'s exposure would be moved from below $k$ to at or above $k$ by the instrument, relative to the number of people in the study population whose exposure would be changed by the instrument. For example, if $\lambda = 0.1$ in the malaria data and we assume the said conditions, $0.1$ is the weighted average reduction in stunting from a one-unit reduction in malaria episodes among individuals who were protected from malaria by the sickle cell trait. Similarly, each weight $w_{ijk}$ represents the $j$th individual in $i$th stratum's  protection from at least $k$ malaria episodes by carrying the sickle cell trait compared to the overall number of individuals who are protected from varying degrees of malaria episodes by carrying the sickle cell trait. In short, the interpretation of $\lambda$ is akin to Theorem 1 in \citet{angrist_two-stage_1995}, except that our result is for the finite-sample case and is specific to matching.

Also, with regards to identification, technically speaking, only assumptions (A1) and (A3) are necessary to identify the `bare-bone' interpretation of $\lambda$ in \eqref{eq:effectRatio}, the ratio of causal effects of the instrument on the outcome (numerator) and on the exposure (denominator) since the numerator and the denominator can both be identified by the differences in expectations in \eqref{eq:effectRatioID}. However, without (A2), i.e. the exclusion restriction, and the monotonicity assumption, this ratio of differences in expectations in \eqref{eq:effectRatioID} cannot identify the weighted average \eqref{eq:effectRatioCACE} of effects of the exposure described in the above paragraph.

When full matching is used so that all subject are used in the matching, the effect ratio \eqref{eq:effectRatio} and its equivalent expression \eqref{eq:effectRatioCACE} are defined for the whole study population. Additionally, the effect ratio is invariant to the particular full match it used. For instance, if a different distance between pairs of subjects were used that resulted in a different full match, the effect ratio would remain the same. Also, one of the advantages of using full matching compared to other matching algorithms that discard some data, such as pair matching, matching with fixed controls, and matching with variable controls, is that full matching estimates the effect ratio \eqref{eq:effectRatio} (or equivalently \eqref{eq:effectRatioCACE}) for the whole study population whereas for the matching methods that discard data, these methods only estimate \eqref{eq:effectRatio} for the data that was not discarded, making the parameter estimate dependent on the individuals that were discarded from the matching algorithm. In contrast, the full matching algorithm incorporates all the individuals in the data and the effect ratio parameter, specifically the subscripts $i,j$ count all the individuals in the data. In fact, the effect ratio \eqref{eq:effectRatio} generalizes previous expressions for the effect ratio with pair matching, $n_i = 2$, by \citet{baiocchi_building_2010} or matching with fixed controls, $n_i = c$, by \citet{kang_causal_2013}. 

\subsection{Inference for effect ratio} \label{sec:inferenceEffectRatio}
We would like to conduct the following hypothesis test for the effect ratio $\lambda$.
\begin{equation} \label{eq:hyp}
H_0: \lambda = \lambda_0, \quad{} H_a: \lambda \neq \lambda_0
\end{equation}
To test the hypothesis in \eqref{eq:hyp}, we propose the following test statistic
\begin{equation} \label{eq:testStat}
T(\lambda_0) = \frac{1}{I} \sum_{i=1}^I V_{i}(\lambda_0) 
\end{equation}
where 
\[
V_{i}(\lambda_0)= \frac{n_i}{m_i} \sum_{j=1}^{n_i} Z_{ij} (R_{ij} - \lambda_0 D_{ij}) - \frac{n_i}{n_i - m_i}\sum_{j=1}^{n_i} (1 - Z_{ij}) (R_{ij} - \lambda_0 D_{ij})
\]
and $S^2(\lambda_0)$, the estimator for the variance of the test statistic, $Var\{T(\lambda_0) | \mathcal{F},\mathcal{Z}\}$ 
\begin{equation} \label{eq:estVar}
S^2(\lambda_0) = \frac{1}{I(I-1)}\sum_{i=1}^I \{V_{i}(\lambda_0) - T(\lambda_0)\}^2
\end{equation} 
Each variable $V_{i}(\lambda_0)$ is the difference in adjusted responses, $R_{ij} - \lambda_0 D_{ij}$, of those individuals with $Z_{ij} = 1$ and those with $Z_{ij} = 0$. Under the null hypothesis in \eqref{eq:hyp}, these adjusted responses have the same expected value for $Z_{ij} = 1$ and $Z_{ij} = 0$ and thus, deviation of $T(\lambda_0)$ from zero suggests $H_0$ is not true. 

Proposition 2 in the supplementary article \citep{kang_supp_2015} states that under regularity conditions, the asymptotic null distribution of $T(\lambda_0)/S(\lambda_0)$ is standard Normal. This provides a point estimate as well as a confidence interval for the effect ratio. For the point estimate, in the spirit of \citet{hodges_estimation_1963}, we find the value of $\lambda$ that maximizes the p-value, Specifically, setting $T(\lambda) / S(\lambda) = 0$ and solving for $\lambda$ gives an estimate for the effect ratio, $\hat{\lambda}$
\[
\hat{\lambda} = \frac{\sum_{i=1}^I \frac{n_i^2}{m_i (n_i - m_i)} \sum_{j=1}^{n_i} (Z_{ij} - \bar{Z}_{i.}) (R_{ij} - \bar{R}_{i.})}{\sum_{i=1}^I \frac{n_i^2}{m_i(n_i - m_i)} \sum_{j=1}^{n_i} (Z_{ij} - \bar{Z}_{i.}) (D_{ij} - \bar{D}_{i.})}
\]
where $\bar{Z}_{i.}, \bar{R}_{i.}$, and $\bar{D}_{i.}$ are averages of the instrument, response, and exposure, respectively, within each matched set. For confidence interval estimation, say 95\% confidence interval, we can solve the equation $T(\lambda)/S(\lambda) = \pm 1.96$ for $\lambda$ to get the confidence interval for the effect ratio. A closed form solution for the confidence interval is provided in Corollary 1 of the supplementary article \citep{kang_supp_2015}.

For our analysis of the malaria data, the regularity conditions, specifically the moment conditions in Proposition 2 of the supplementary article \citep{kang_supp_2015} (i.e.  $V_{i}^4(\bar{\lambda})$ is uniformly bounded), are automatically met because the responses are binary (i.e. stunted or not stunted) and the malaria episodes are bounded whole numbers. Hence, Proposition 2 and its subsequent Corollary 1 from the supplementary article \citep{kang_supp_2015} are used to compute the point estimate, the p-value, and the confidence intervals for the casual effect of malaria on stunting. Note that the inferences we develop for the effect ratio allow for non-binary outcomes and exposures, even though our malaria data have binary outcomes and whole-number exposures. 

\subsection{Sensitivity analysis} \label{sec:sens}
Sensitivity analysis attempts to measure the influence of unobserved confounders on the inference on $\lambda$. In the case of instrumental variables, a sensitivity analysis quantifies how a violation of assumption (A3) in Section \ref{sec:defIV} would impact the inference on $\lambda$ \citep{rosenbaum_observational_2002}. Specifically, under assumption (A3), the instrument is assumed to be free from unmeasured confounders or free after conditioning on observed confounders via matching. The latter implies that the instruments are assigned randomly, $P(\mathbf{Z} = z | \mathcal{F}, \mathcal{Z}) = (|\Omega|)^{-1}$, i.e. that within each matched set $i$, $P(Z_{ij} = 1 | \mathcal{F}, \mathcal{Z}) = m_i / n_i$. 

However, as discussed in Section \ref{sec:defIV}, even after matching for observed confounders, unmeasured confounders may influence the viability of assumption (A3). For example, with the malaria study, within a matched set $i$ , two children, $j$ and $k$, may have the same birth weights, be from the same village, and have the same covariate values ($\mathbf{x}_{ij} = \mathbf{x}_{ik}$), but have different probabilities of carrying the HbAS genotype, $P(Z_{ij} = 1 | \mathcal{F}) \neq P(Z_{ik}=1 | \mathcal{F})$ due to unmeasured confounders, denoted as $u_{ij}$ and $u_{ik}$ for the $j$th and $k$th unit, respectively. Despite our best efforts to minimize the observed differences in covariates and to adhere to assumption (A3) after conditioning on the matched sets, unmeasured confounders such as a child's family's ancestry could still be different between the $j$th and $k$th child, and this difference could make the inheritance of the sickle cell trait depart from randomized assignment, violating assumption (A3).

To model this deviation from randomized assignment due to unmeasured confounders, let $\pi_{ij} = P(Z_{ij} = 1 | \mathcal{F})$ and $\pi_{ik} = P(Z_{ik} = 1 | \mathcal{F})$ for each unit $j$ and $k$ in the $i$th matched set. The odds that unit $j$ will receive $Z_{ij} =1$ instead of $Z_{ij} = 0$ is $\pi_{ij}/(1 - \pi_{ij})$. Similarly, the odds for unit $k$ is $\pi_{ik}/(1 - \pi_{ik})$. Suppose the ratio of these odds is bounded by $\Gamma \geq 1$
\begin{equation} \label{eq:gamma}
\frac{1}{\Gamma} \leq \frac{\pi_{ij}  (1 -\pi_{ik})}{\pi_{ik} (1 - \pi_{ij})} \leq \Gamma
\end{equation}
If unmeasured confounders play no role in the assignment of $Z_{ij}$, then $\pi_{ij} = \pi_{ik}$ and $\Gamma = 1$. That is, child $j$ and $k$ have the same probability of receiving $Z_{ij} =1$ in matched set $i$. If there are unmeasured confounders that affect the distribution of $Z_{ij}$, then $\pi_{ij} \neq \pi_{ik}$ and $\Gamma > 1$. For a fixed $\Gamma > 1$, we can obtain lower and upper bounds on $\pi_{ij}$, which can be used to derive the null distribution of $T(0)/S(0)$ under $H_0: \lambda = 0$ in the presence of unmeasured confounding and be used to compute a range of possible p-values for the hypothesis $H_0: \lambda = 0$ \citep{rosenbaum_observational_2002}. The range of p-values indicates the effect of unmeasured confounders on the conclusions reached by the inference on $\lambda$. If the range contains $\alpha$, the significance value, then we cannot reject the null hypothesis at the $\alpha$ level when there is an unmeasured confounder with an effect quantified by $\Gamma$. In addition, we can amplify the interpretation of $\Gamma$ using \citet{rosenbaum_amplification_2009} to get a better understanding of the impact of the unmeasured confounding on the outcome and the instrument (see the supplementary article \citep{kang_supp_2015} for the derivation of the sensitivity analysis and the amplification of $\Gamma$) 

\section{Simulation Study} 
\subsection{Robustness of our method} \label{sec:simulation.2sls}
One of the advantages of matching based IV estimation versus traditional IV estimation, such as conventional 2SLS without matching, is its robustness to parametric assumptions between the outcome and the covariates. Specifically, for conventional 2SLS, in order for the estimate to be consistent, the covariates must have a linear effect on the expected outcome. In contrast, matching-based IV estimation puts no constraints on the structure of the relationship between the outcome and the covariates. In this section, we study this phenomena in detail through a simulation study.

Let the outcome $R_{ij}$, the exposure $D_{ij}$, the observed covariates $\mathbf{X}_{ij}$, and the instrument $Z_{ij}$ be generated based on the following model known as the structural equations model in econometrics \citep{wooldridge_econometrics_2010}.
\begin{align*}
\left.
\begin{array}{r@{\mskip\thickmuskip}l}
R_{ij} &=  \alpha + \beta D_{ij} + f(\mathbf{X}_{ij}) + \epsilon_{ij} \\
D_{ij} &= \kappa + \pi Z_{ij} + \bm{\rho}^T \mathbf{X}_{ij} + \xi_{ij}
\end{array}
\quad{},
\begin{array}{r@{\mskip\thickmuskip}l}
\begin{pmatrix}
\epsilon_{ij} \\
\xi_{ij} 
\end{pmatrix} &\iid N \left( \begin{bmatrix} 0 \\ 0 \end{bmatrix}, \begin{bmatrix} 1 & 0.8 \\ 0.8 & 1\end{bmatrix} \right)
\end{array}\right.
\end{align*}
where the parameters $\alpha,\beta,\kappa$ and $\bm{\rho}$ are all fixed throughout the simulation. The parameters $\alpha$ and $\kappa$ are intercepts. The parameter $\beta$ is the quantity of interest, the effect of the exposure on the outcome, and is also equal to the effect ratio (see Section 1 of the supplementary article \citep{kang_supp_2015} for details). The parameter $\pi$ quantifies the strength of the instrument. The function $f(\cdot)$ is a pre-defined function that takes in a vector of observed covariates $\mathbf{X}_{ij}$ and produces a scalar value that affects the outcome, $R_{ij}$. In the simulation, $\mathbf{X}_{ij}$, are five-dimensional vectors or $\mathbf{X_{ij}} = (X_{ij1},\ldots,X_{ij5})$. Also, we consider the following list of functions parametrized by $\bm{\gamma} \in \reals^5$
\begin{enumerate}
\item[(a)] Linear function: $f(\mathbf{X}_{ij}) = \sum_{k=1}^5 \gamma_k X_{ijk}$
\item[(b)] Quadratic function: $f(\mathbf{X}_{ij}) = \sum_{k=1}^5 \gamma_k X_{ijk}^2  $
\item[(c)] Cubic function: $f(\mathbf{X}_{ij}) = \sum_{k=1}^5 \gamma_{k} X_{ijk}^3$
\item[(d)] Exponential function: $f(\mathbf{X}_{ij}) = \sum_{k=1}^5 \gamma_{k} \exp(X_{ijk})$
\item[(e)] Log function: $f(\mathbf{X}_{ij}) = \sum_{k=1}^5 \gamma_{k} \log(|X_{ijk}|)$
\item[(f)] Logistic function: $f( \mathbf{X}_{ij}) = \frac{1}{1 + \exp(- \sum_{k=1}^5 X_{ijk} \gamma_k )}$
\item[(g)] Truncated function: $f(\mathbf{X}_{ij}) = \sum_{k=1}^{5} \gamma_k \chi(X_{ijk} \geq 0)$ where $\chi(\cdot)$ is an indicator function.
\item[(h)] Square root function: $f(\mathbf{X}_{ij}) = \sum_{k=1}^5 \gamma_k \sqrt{|X_{ijk}|}$
\end{enumerate}
To generate $\mathbf{X}_{ij}$, we adopt the following scheme. For individuals with $Z_{ij} =0$, $\mathbf{X}_{ij}$ comes from a five-dimensional multivariate Normal distribution with mean $(0,\ldots,0)$ and an identity covariance matrix. For individuals with $Z_{ij} = 1$, $\mathbf{X}_{ij}$ comes from a five-dimensional multivariate Normal with mean $(1,0,\ldots,0)$ and an identity covariance matrix. The instruments, $Z_{ij}$, are generated randomly with $P(Z_{ij}=1) = 1/8$ and $P(Z_{ij} = 0) = 7/8$, similar to that observed in our malaria data. For each generated data set, we compute the estimate of $\beta$ using 2SLS and our procedure. 2SLS is based on (i) regressing $D_{ij}$ on $Z_{ij}$ and $X_{ij}$ to obtain the predicted value of $D_{ij}$, say $\hat{D}_{ij}$, and (ii) regressing $R_{ij}$ on $\hat{D}_{ij}$ and $\mathbf{X}_{ij}$. We simulate this process 5000 times and compute the estimates of $\beta$ produced by the two procedures. We measure the performance of the two procedures by computing the median absolute deviation, the absolute bias of the median (i.e. the absolute value of the bias of the median estimate with respect to $\beta$), and the Type 1 error rate over 5000 simulations. For each simulation study, we vary the function $f(\cdot)$ and $\pi$.

\begin{figure} 
\includegraphics[scale=0.55]{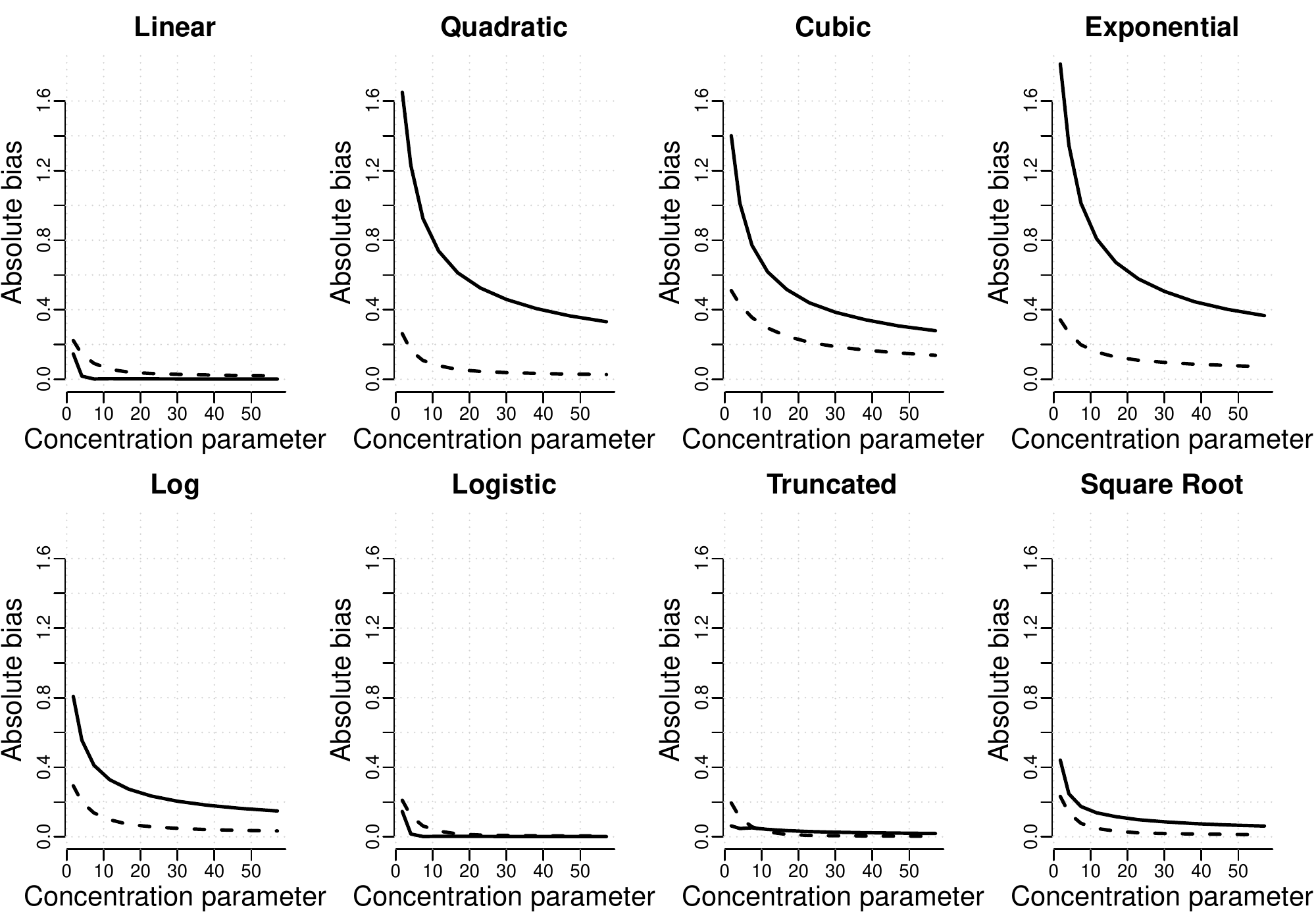}
\caption[]{Absolute bias of the median for our method vs. two stage least squares (2SLS) for different concentration parameters. The solid line indicates 2SLS and the dashed line indicates our method.}
\label{fig:strengthBias}
\end{figure}

Figures \ref{fig:strengthBias} and \ref{fig:strengthTypeI} compare performances between 2SLS and our method when we fix the sample size, but vary the strength of the instrument (i.e. the strength of the effect of the instrument on the treatment) via $\pi$. Specifically, we evaluate the strength of the instrument using a popular measure known as the concentration parameter \citep{bound_problems_1995}. High values of the concentration parameter indicate a strong instrument while low values of it indicate a weak instrument. The concentration parameter is the population value of the first stage partial F statistic for the instruments when the treatment is regressed on the instrument and the measured covariates $\mathbf{X}_{ij}$; this first stage F statistic is often used to check instrument strength where an F below $10$ suggests that the instruments are weak \citep{stock_survey_2002}. The sample size is fixed at $800$ where $100$ individuals have $Z_{ij} = 1$ and $700$ individuals have $Z_{ij} = 0$, similar to the sample size presented in the malaria data. We also vary $f(\cdot)$ based on the functions listed in the previous paragraph.

Figure \ref{fig:strengthBias} measures the absolute bias of the median for 2SLS and our method. When $f(\cdot)$ is a linear function of the observed covariates $\mathbf{x}_{ij}$, 2SLS does slightly better than our method. 2SLS doing well for the linear function is to be expected since 2SLS is consistent when the model is linear. However, if $f(\cdot)$ is non-linear, our matching estimator does better than 2SLS and is never substantially worse for all instrument strengths. For example, for quadratic, cubic, exponential, log, and square root functions, our method has lower bias than 2SLS for all strengths of the instrument. For logistic and truncated functions, our method is similar in performance to 2SLS for all strengths of the instrument. In the supplementary article \citep{kang_supp_2015}, we also measure the median absolute deviation of 2SLS and our method and we find that the price we pay for lower bias of our method in a slight increase in dispersion compared to 2SLS.

\begin{figure}
\includegraphics[width=\textwidth]{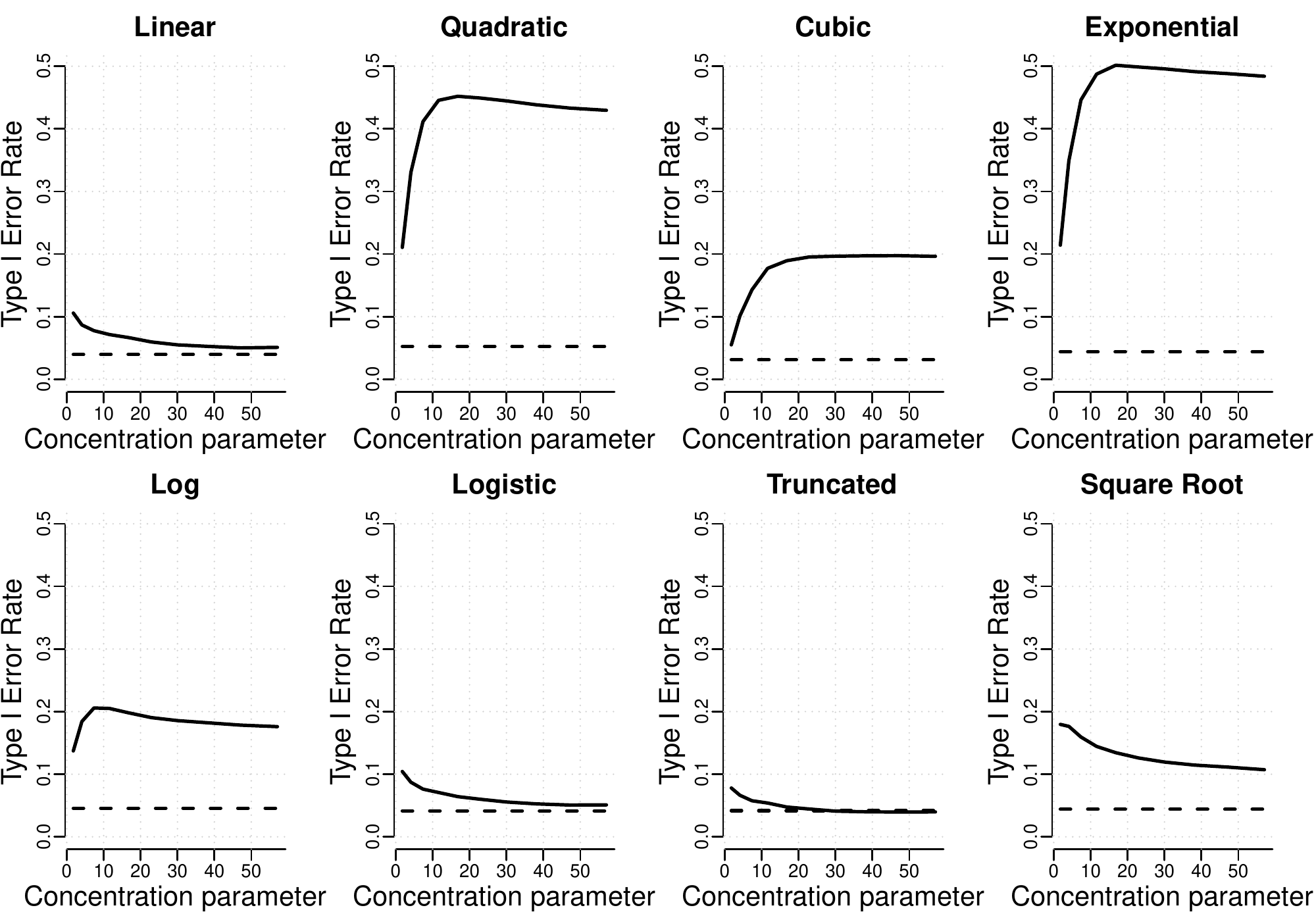}
\caption[]{Type I error rate for our method vs. two stage least squares (2SLS) for different concentration parameters. The solid line indicates 2SLS and the dashed line indicates our method.}
\label{fig:strengthTypeI}
\end{figure}
Finally, Figure \ref{fig:strengthTypeI} measures the Type I error rate of 2SLS and our method. Regardless of the function type and the instrument strength, our method retains the nominal $0.05$ rate. In fact, even for the linear case where 2SLS is designed to excel, our estimator has the correct Type I error rate for all instrument strengths while 2SLS has higher Type I error for weak instruments. For all the non-linear functions, the Type I error rate for 2SLS remains above the 0.05 line while our estimator maintains the nominal Type I error rate. This provides evidence that our estimator will have the correct 95\% coverage for confidence intervals regardless of non-linearity or instrument strength.

In summary, the simulation study shows promise that our method is generally more robust to assumptions about instrument strength and linearity between the outcome and the covariates than 2SLS at the expense of a small increase in dispersion.

\subsection{Comparison to \citet{frolich_nonparametric_2007}} \label{sec:simulation.frolich}

In addition to comparing our method against the most popular IV estimator, 2SLS, we also compare our method to the non-parametric IV method of \citet{frolich_nonparametric_2007} implemented by \citet{frolich_estimation_2010}. The simulation setup is identical to Section \ref{sec:simulation.2sls}, except that we discretize the exposure value $D_i$ so that we can compare our method to the method in \citet{frolich_nonparametric_2007}. Specifically, let $D_{ij}^*$ be defined as $D_{ij}$ in Section \ref{sec:simulation.2sls}, i.e. $D_{ij}^* = \kappa + \pi Z_{ij} + \mathbf{\rho}^T \mathbf{X}_{ij} + \xi_{ij}$. Then, we define
\[
D_{ij} = \chi(D_{ij}^* < -1) + 2\chi(-1 \leq D_{ij}^* < 1) + 3\chi(1 \leq D_{ij}^*)
\]
where $\chi(\cdot)$ is the indicator function. The response $R_{ij}$ is generated from the same model as in Section \ref{sec:simulation.2sls}, except with a discretetized $D_{ij}$. The rest of the data generating process is identical to Section \ref{sec:simulation.2sls}. 

For each simulated data, we use the code provided by \citet{frolich_estimation_2010} to generate an estimate for $\beta^*$, the local average treatment effect, with the default settings for the tuning parameters. We also use our method to estimate $\beta^*$. Finally, for comparison, we run 2SLS on the simulate data. As before, we measure the median absolute deviation and the absolute bias of the median. For each simulation study, we vary the function $f(\cdot)$ and $\pi$, the strength of the instrument.

Figures \ref{fig:frolichBias} and \ref{fig:frolichMAD} show the absolute bias and median absolute deviation between the three methods. Generally speaking, both our method and \citet{frolich_nonparametric_2007}'s method do better than 2SLS when $f(\cdot)$ is non-linear. Between our method and \citet{frolich_nonparametric_2007}'s method, in most cases, our method is better or similar to the \citet{frolich_nonparametric_2007}'s method when it comes to bias. With regards to variability, our method and \citet{frolich_nonparametric_2007}'s method are very similar to each other. For the quadratic, cubic, and exponential functions, our simulations show that our method dominates both in bias and variance compared to \citet{frolich_nonparametric_2007}. Further details of the simulation in this Section can be found in the supplementary article \citep{kang_supp_2015}.

\begin{figure}
\includegraphics[width=\textwidth]{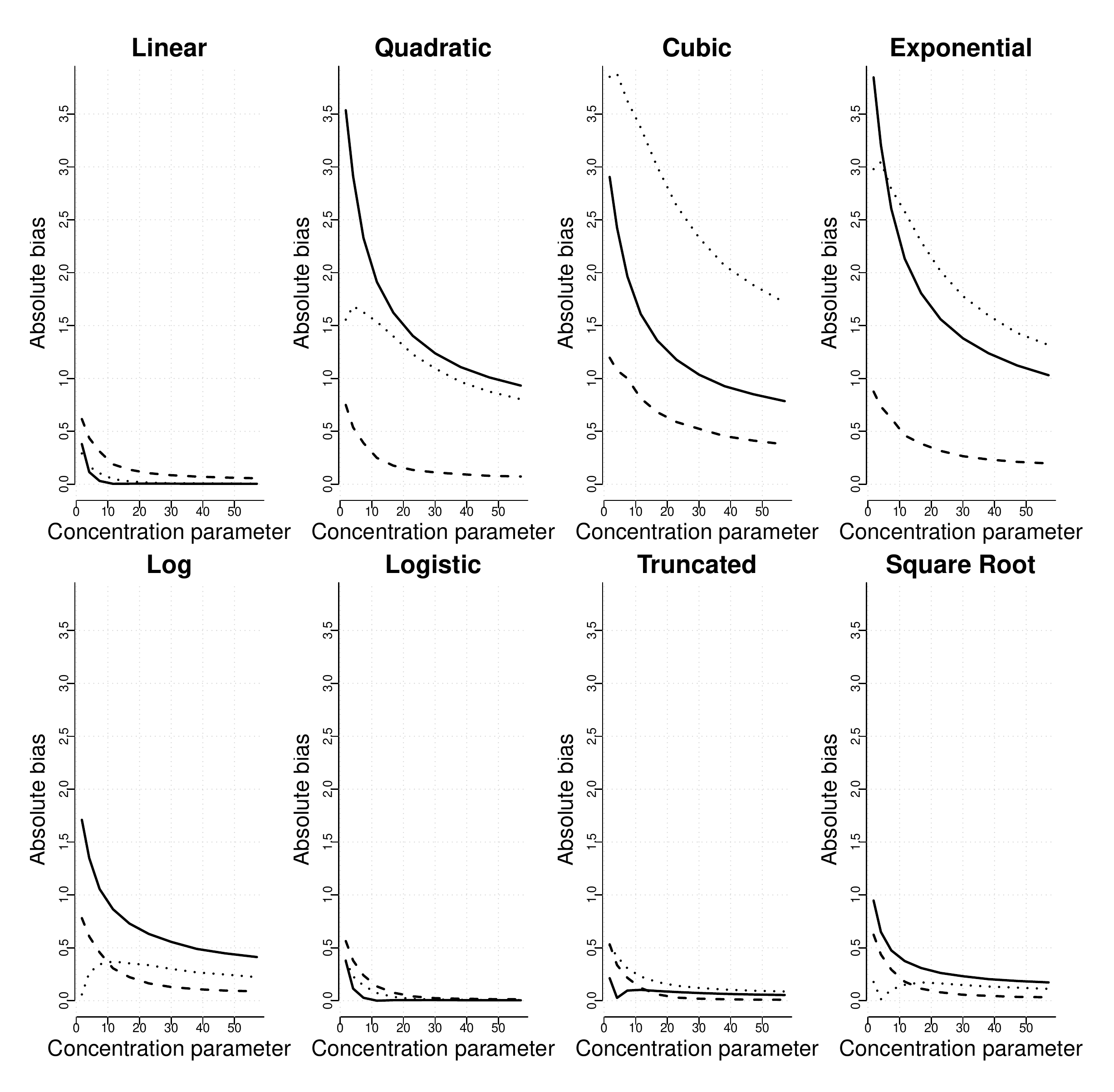} 
\caption[]{Absolute bias of the median for our method, 2SLS, and Fr\"{o}lich's method for different concentration parameters. The solid line indicates 2SLS, the dashed line indicates our method, and the dotted line indicates Fr\"{o}lich's method.}
\label{fig:frolichBias}
\end{figure}

\begin{figure}
\includegraphics[width=\textwidth]{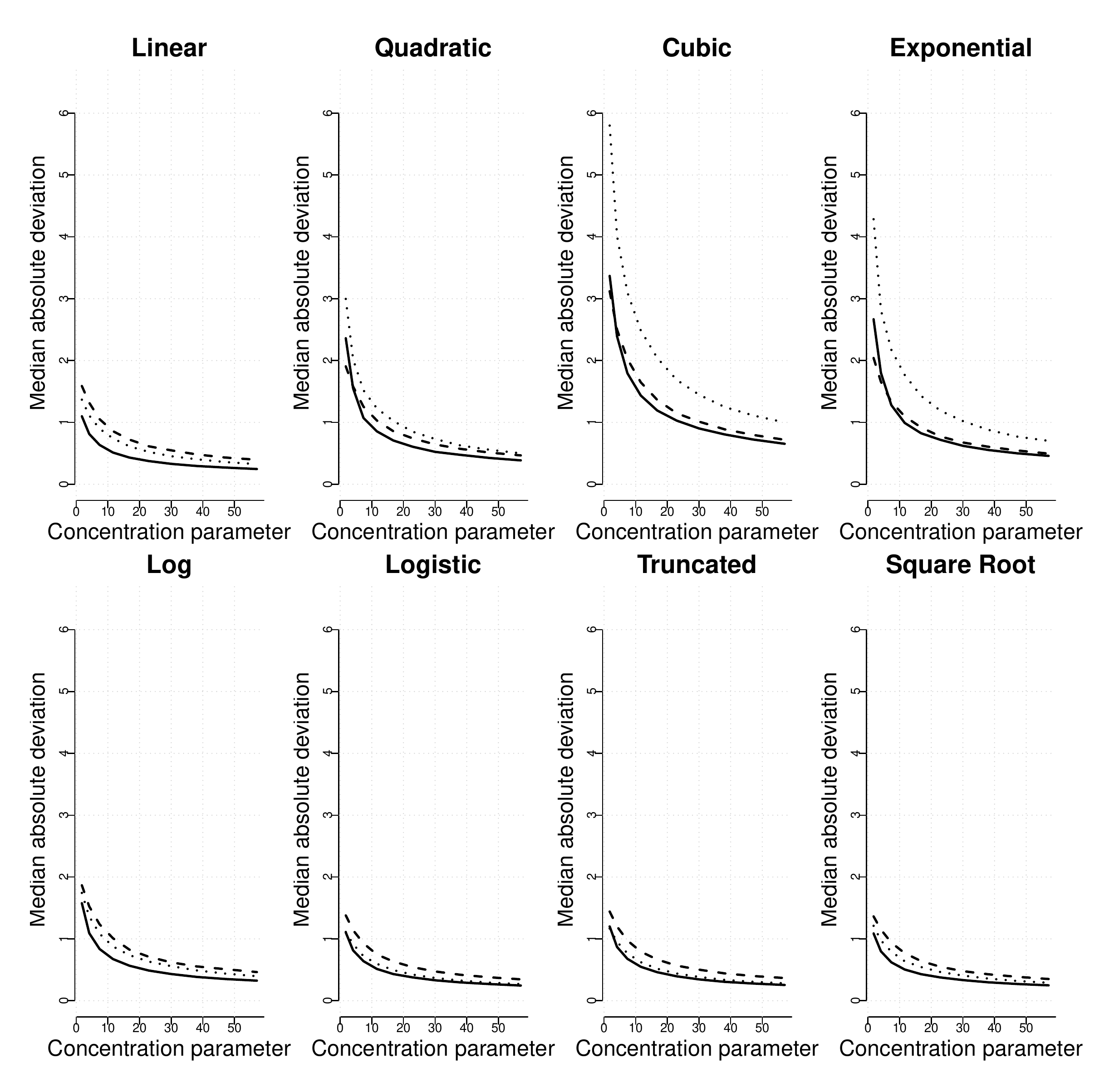}
\caption[]{Median absolute deviation between our method, 2SLS, and Fr\"{o}lich's method for different concentration parameters. The solid line indicates 2SLS, the dashed line indicates our method, and the dotted line indicates Fr\"{o}lich's method.}
\label{fig:frolichMAD}
\end{figure}

\section{Data Analysis of the Causal Effect of Malaria on Stunting} \label{sec:data}
\subsection{Background information} \label{sec:data.background}
Using the new full matching IV method in this paper, we analyze the data set introduced in Section \ref{sec:intro.motivatingExample} to study the causal effect of malaria on stunting. Following \citet{kreuels_sickle_2009}, we only consider children with the heterozygous strand HbAS, the sickle cell trait, or wildtype HbAA and exclude children with the homozygous strand (HbSS), or a different mutation on the same gene leading to hemoglobin C (HbAC, HbCC, HbSC); this reduced the sample size from $1070$ to $884$. Among $884$ children, $110$ children carried HbAS and $774$ children carried HbAA. 

The instrument was a binary variable indicating either the HbAS or HbAA genotype. The exposure of interest was the malarial history, which was defined as the total number of malarial episodes during the study. A malaria episode was defined as having a parasite density of more than 500 parasites/$\mu$l and a body temperature greater than 38$^{\circ}$C or the mother reported a fever within the last 48 hours. The outcome was whether the child was stunted at the last recorded visit, which occurred when the child was approximately two years old. The difference in episodes of malaria between children with HbAS and HbAA is significant (Risk ratio: 0.82, p-value: 0.02, 95\% CI: (0.70, 0.97)), indicating that the sickle cell trait instrument satisfies Assumption (A1) of being associated with the exposure. 

\begin{sidewaystable}
\caption{Characteristics of study participants at recruitment. P-values were obtained by doing a Pearson's chi-squared test for categorical covariates and two-sample t tests for numerical variables. *** corresponds to a p-value of less than 0.01, ** corresponds to a p-value between 0.01 and 0.05, and * corresponds to a p-value between 0.05 and 0.1.}
\label{tab:char}
\begin{tabular}{p{6.5cm}  l l}
& HbAS ($n=110$) & HbAA ($n=774$)  \\ \cline{2-3} \\
Birth weight (Mean,(SD))& 3112.44 (381.9) (32 missing) & 2978.7 (467.9) (239 missing) *** \\
Sex (Male/Female) & 46.4\% Male & 51.0\% Male \\
Birth season (Dry/Rainy) & 56.4\% Dry & 55.3\% Dry \\
Ethnic group  (Akan/Northerner) & 86.4\% Akan  & 88.8 \% Akan (4 missing) \\
$\alpha$-globin genotype (Norm/Hetero/Homo) & 75.7\% / 21.5\% /  2.8\% (3 missing) & 74.4\% / 23.1\% / 2.6\% (29 missing)\\
Village of residence: & & \\
~~  Afamanso &4.6 \% & 4.8\% \\
~~  Agona &10.0\% & 13.6\% \\
~~  Asamang& 13.6\% & 11.1\% \\
~~  Bedomase & 5.5\% & 4.5\% \\
~~  Bipoa & 14.5 \% & 10.7\%  \\
~~  Jamasi & 15.5 \% & 13.8\%  \\
~~  Kona & 16.4 \% & 12.8\%  \\
~~  Tano-Odumasi & 4.5 \% & 12.3\%** \\
~~  Wiamoase &15.5 \% & 16.4\%  \\ 
Mother's occupation (Nonfarmer/Farmer) & 79.0\% Nonfarmer & 78.0\% Nonfarmer (11 missing) \\
Mother's education (Literate/Illiterate) & 91.7\% Literate (2 missing) & 90.5\% Literate (8 missing) \\
Family's financial status (Good/Poor) & 69.1\% Good (13 missing)& 70.1\% Good  (84 missing)\\ 
Mosquito protection (None/Net/Screen) & 55.7\% / 32.0\% / 12.4\% (13 missing) & 45.4\%* / 35.1\% / 19.5\%  (76 missing) \\ 
Sulphadoxine pyrimethamine (Placebo/SP) & 49.1\% Placebo & 50.1\% Placebo \\
\end{tabular}
\end{sidewaystable}

Table \ref{tab:char} summarizes all the measured covariates in this data. All the covariates were collected at the beginning of the study, which is three months after the child's birth. We will match on all these covariates for reasons that will be explained below. Broadly speaking, for valid inference of the causal effect using instrumental variables, we would like to include all confounders for the instrument-outcome relationship, i.e. covariates that are determined before (or at the same time and not affected by) the sickle cell trait and that are associated with the outcome. The following covariates in Table \ref{tab:char}, village of residence, sex, ethnicity, birth season, and alpha-globin genotype, represent such potential confounders. They occur before (or at the same time and are not affected by) the sickle cell trait and they could be associated with the outcome of stunting through population stratification. If these covariates were the only instrument-outcome confounders, then we would not need to consider matching for other covariates. 

However, other possible confounders in our data include family's socioeconomic status and parents' sickle cell genotype. Family's socioeconomic status may be associated with the sickle cell trait through population stratification and can affect the outcome of stunting through the nutrition and hygienic environment of the child. Parents' sickle cell genotype is associated with the child's sickle cell genotype because of the properties of genetic inheritance and may be associated with the outcome of stunting through population stratification. Although these two possible confounders were not measured at the time of instrument assignment (i.e. the child's conception), the following covariates in Table \ref{tab:char}, birthweight, mother's occupation, mother's education, family's financial status, and mosquito protection are proxies for these variables. Specifically, mother's occupation, mother's education, and family's financial status measured three months after the child's birth are proxies of family's socioeconomic status at the time of the child's conception. Mosquito protection and birthweight are proxies for parents' sickle cell genotype. In particular, mosquito protection at the time of the child's conception (i.e. whether the family's home is protected by nets, screens, or nothing) may be associated with parents' sickle cell genotype because a family might be less likely to seek additional mosquito protection if members of the family are naturally protected by being carriers of the sickle cell genotype; one can see in Table \ref{tab:char} that children carrying HbAS tend to have less mosquito protection than child carrying HbAA. Birthweight may be associated with maternal sickle cell genotype because a mother having HbAS may be protected from malaria during pregnancy, which may increase birthweight \citep{eisele_malaria_2012}. 

But, matching on covariates that are measured or determined after the instrument such as birthweight, mosquito protection, and family's socioeconomic status three months after the child's birth could create bias if the instrument could alter these values \citep{rosenbaum_consequences_1984}. However, we think the child's sickle cell trait instrument does not alter these covariates because children are generally protected from malaria in the first three months of life due to maternal antibodies \citep{snow_risk_1998} and parents were generally not aware of the child's sickle cell genotype. Consequently, the child's sickle cell genotype does not affect the child's birthweight, family decisions about mosquito protection, and family's socioeconomic status at the time the child is three months old and these variables are effectively pre-instrument covariates so that matching for them does not create bias \citep{rosenbaum_consequences_1984,holland_statistics_1986}. In short, we match for all the covariates in Table \ref{tab:char} because they are either pre-instrument potential confounders or effectively pre-instrument proxies for unmeasured potential confounders.

Finally, we note that some of the covariates in Table \ref{tab:char} may not be highly associated with the sickle cell trait genotype. For example, sulpadoxine pyrimethamine vs. placebo was randomly assigned as part of a randomized trial. However, we still have chosen to match on all the covariates because each covariate may be associated with the outcome and matching a covariate that is associated with the outcome increases efficiency and reduces sensitivity to unobserved biases \citep{rosenbaum_heterogeneity_2005, zubizarreta_matching_2014}. Furthermore,  \citet{rubin_should_2009} argues for erring on the side of being inclusive when deciding which variables to match on (i.e. control for) in an observational study. Failure to match for a covariate that has an important effect on outcome and is slightly out of balance can cause substantial bias. 

In terms of the balance of the covariates in Table \ref{tab:char}, before matching, we see that there are a few significant differences between the HbAS and HbAA groups, most notably in birth weight, village of birth, and mosquito protection status. Children with the sickle cell trait (HbAS) tend to have high birth weights and lack any protection against mosquitos compared to HbAA children. Also, children living in the village of Tano-Odumasi tend to inherit HbAA more frequently than HbAS. Any one of these differences can contribute to the violation of IV assumption (A3) in Section \ref{sec:defIV} if we do not control for these differences. For instance, it is possible that children with low birth weights were malnourished at birth, making them more prone to malarial episodes and stunted growth compared to children with high birth weights. We must control for these differences to eliminate this possibility, which we do through full matching. 

\subsection{Implementation of full matching on data} \label{sec:data.match}
We conduct full matching on all observed covariates. In particular, we group children with HbAS and without HbAS based on all the observed characteristics in Table \ref{tab:char} as well as match for patterns of missingness. To measure similarity of the observed and missing covariates, we use the rank-based Mahalanobis distance as the distance metric for covariate similarity \citep{rosenbaum_design_2010}. In addition, we compute propensity scores by logistic regression. Here, the propensity score is an instrumental propensity score, which is the probability of having the sickle cell trait given the measured confounders \citep{cheng_using_2011}. In addition, children with missing values in their covariates were matched to other children with similar patterns of missing data \citep{rosenbaum_design_2010}. Once covariate similarity was calculated, the matching algorithm optmatch in R \citep{hansen_optimal_2006} matched children carrying HbAS with children carrying HbAA in a way that within each matched set, their covariates are similar.

\begin{figure}
\includegraphics[scale = 0.65]{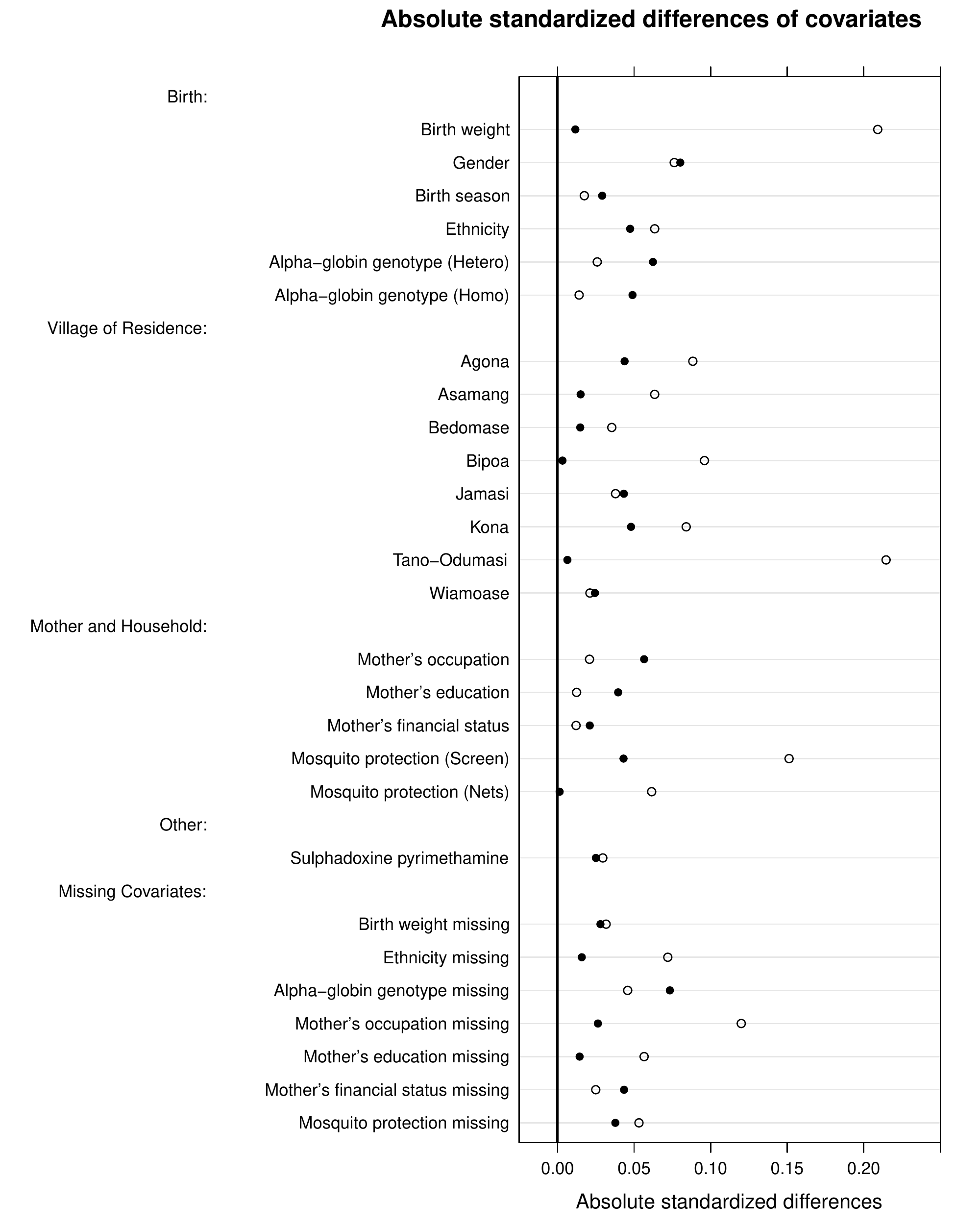}
\caption[]{Absolute standardized differences before and after full matching. Unfilled circles indicate differences before matching and filled circles indicate differences after matching.}
\label{fig:fullMatch0}
\end{figure}

Figure \ref{fig:fullMatch0} shows covariate balance before and after full matching using absolute standardized differences. Absolute standardized differences before matching are computed by taking the difference of the means between children with HbAS and HbAA for each covariate, taking the absolute value of it, and normalizing it by the within group standard deviation before matching (the square root of the average of the variances within the groups). Absolute standardized differences after matching are computed by taking the  differences of the means between children with HbAS and HbAA within each strata, averaging this difference across strata, taking the absolute value of it, and normalizing it by the same within group standard deviation before matching as before. Before matching, there are differences in birth weight, mosquito protection, and village of residence between children with HbAS and HbAA. After matching, these covariates are balanced. Specifically, the standardized differences for birth weight, village of residence, and mosquito protection, are under $0.1$ indicating balance \citep{normand_validating_2001}. In fact, all the covariates are balanced after matching and the p-values used to test the differences between HbAS and HbAA in Table \ref{tab:char} are no longer significant after matching. \citet{hansen_full_2004} discusses how the size of matched sets in full matching can be restricted. In the supplementary article \citep{kang_supp_2015}, we compare different restrictions on full matching versus unrestricted full matching in terms of balance and efficiency. In short, the analysis reveals that unrestricted full matching creates the most covariate balance by a substantial amount while having a only slight decrease in efficiency compared to other full matching schemes considered and hence, we use unrestricted full matching.

\subsection{Estimate of causal effect of malaria on stunting} \label{sec:data.est}
\begin{table}
\caption{Estimates of the causal effect using full matching compared to two-stage least squares and multiple regression. }
\label{tab:estData} 
\begin{tabular}{l l l l}
Methods & Estimate & P-value & 95\% confidence interval \\ \hline
Our method & 0.22 & 0.011 & (0.044, 1) \\
Two stage least squares & 0.21 & 0.14 & (-0.065, 0.47) \\
Multiple regression & 0.018 & 0.016 & (0.0034, 0.033)
\end{tabular}
\end{table}
Table \ref{tab:estData} shows the estimates of the causal effect of malaria on stunting from different methods, specifically our matching-based method, conventional two stage least squares (2SLS), and multiple regression. Our matching-based method computed the estimate by the procedure outlined in Section \ref{sec:inferenceEffectRatio}. 2SLS computed the estimate by regressing all the measured covariates and the instrument on the exposure and using the prediction from that regression and the measured covariates to obtain the estimated effect. Inference for 2SLS was derived using standard asymptotic Normality arguments \citep{wooldridge_econometrics_2010}. Finally, the multiple regression estimate was derived by regressing the outcome on the exposure and the covariates and the inference on the estimate was based on a standard t test. 

We see that the full matching method estimates $\lambda$ to be $0.22$. That is, the risk of stunting among children with the sickle cell trait is estimated to decrease by $0.22$ times the average malaria episodes prevented by the sickle cell trait. Furthermore, we reject the hypothesis $H_0: \lambda = 0$, that malaria does not cause stunting, at the $0.05$ significance level. The confidence interval $\lambda$ is $(0.044, 1.0)$. Even the lower limit of this confidence interval of $0.044$ means that malaria has a substantial effect on stunting; it would mean that the risk of stunting among children with the sickle cell trait is decreased by $0.044$ times the average malaria episodes prevented by the sickle cell trait. 

The estimate based on 2SLS is 0.21, similar to our method. However, our method achieves statistical significance but 2SLS does not. Also, multiple regression, which does not control for unmeasured confounders, estimates a much smaller effect of 0.018. 
\begin{table}
\caption{Sensitivity analysis for instrumental variables with full matching. The range of significance is the range of p-values over the different possible distributions of the unmeasured confounder given a particular value of $\Gamma$, which represents the effect of unobserved confounders on the inference of $\lambda$.}
\label{tab:sensData} 
\begin{tabular}{l  l}
Gamma & Range of significance\\ \hline
1.1 & (0.0082, 0.041) \\
1.2 & (0.0034, 0.074)  \\
1.3 & (0.0015, 0.12)  
\end{tabular}
\end{table}

Table \ref{tab:sensData} shows the sensitivity analysis due to unmeasured confounders. Specifically, we measure how sensitive our estimate and the p-value in Table \ref{tab:estData} is to violation of assumption (A3) in Section \ref{sec:defIV}, even after matching. We see that our results are somewhat sensitive to unmeasured confounders at the $0.05$ significance level. If there is an unmeasured confounder that increases the odds of inheriting HbAS over HbAA by 10\%, i.e. $\Gamma = 1.1$, then we would still have strong evidence that malaria causes stunting. But, if an unmeasured confounder increases the odds of inheriting HbAS over HbAA in a child by 20\% (i.e. $\Gamma = 1.2$), the range of possible p-values includes $0.05$, the significance level, meaning that we would not reject the null hypothesis of $H_0: \lambda = 0$, that malaria does not cause stunting. In the supplementary article \citep{kang_supp_2015}, we amplify the sensitivity analysis following \citet{rosenbaum_amplification_2009}.

\section{Summary}

Overall, in contrast to regression-based IV estimation procedures like 2SLS, our full matching IV method (i) provided a clear way to assess the balance of observed covariates and design the study without looking at the outcome data and (ii) provided a method to quantify the effect of unmeasured confounders on our inference of the causal effect. Our method made it explicitly clear how these covariates were adjusted by stratifying individuals based on similar covariate values. Finally, like in a randomized experiment, our analysis only looked at the outcome data once the balance was acceptable, i.e. once the differences in birth weight, village of residence, and mosquito protection between children with HbAS and HbAA were controlled for. If the balance was unacceptable, then comparing the outcomes between the two groups would not provide reliable causal inference since any differences in the outcome can be attributed to the differences in the covariates. In contrast, conventional 2SLS can only analyze the causal relationship in the presence of outcome data, making the outcome data necessary throughout the entire analysis. Finally, our method is robust to parametric modeling assumptions between the outcome and the covariates with respect to Type I error and point estimate, which cannot be said about 2SLS.

At the expense of these benefits, especially blinding and transparency with regards to covariate balance, unfortunately matching estimators tend to be less efficient than 2SLS or some of the semiparametric methods mentioned in Section \ref{sec:intro.tsls} when the semiparametric methods' assumptions hold. In practice, our estimator's blinding and transparency can be a powerful design and visual tool for applied researchers to assess the validity of the causal conclusions. However, a more careful exploration of the trade-offs between the efficiency of our estimator and the efficiency of some of the semiparametric and non-parametric methods is an interesting direction for future research.

\bibliographystyle{imsart-nameyear}
\bibliography{mainbib}

\appendix

\section{Supplementary Materials: Review of Notation}
We adopt the notation in Section 2.1 of the main manuscript. Also, we define the effect ratio, $\lambda$, our estimator $T(\lambda_0)$, and $S^2(\lambda_0)$ identically as Sections 2.4 and 2.5 of the main text.
\begin{align} \label{eq:effectRatioGeneral}
\lambda &= \frac{\sum_{i=1}^I \sum_{j=1}^{n_i} r_{1ij}^{(d_{1ij})} - r_{0ij}^{(d_{0ij})}}{\sum_{i=1}^I \sum_{j=1}^{n_i} d_{1ij} - d_{0ij}} \\
\label{eq:genTestStat}
T(\lambda_0) &= \frac{1}{I} \sum_{i=1}^I V_i(\lambda_0)\\
 \label{eq:genEstVar}
S^2(\lambda_0) &= \frac{1}{I(I-1)}\sum_{i=1}^I \{V_{i}(\lambda_0) - T(\lambda_0)\}^2
\end{align}
where 
\[
V_{i}(\lambda_0)= \frac{n_i}{m_i} \sum_{j=1}^{n_i} Z_{ij} (R_{ij} - \lambda_0 D_{ij}) - \frac{n_i}{n_i - m_i}\sum_{j=1}^{n_i} (1 - Z_{ij}) (R_{ij} - \lambda_0 D_{ij})
\]

\section{Supplementary  Materials: Identification and Interpretation of the Effect Ratio}
Let $\chi(\cdot)$ be an indicator function. Under the IV assumptions laid out in Section 2.3 of the main paper and the monotonicity assumption where $d_{1ij} \geq d_{0ij}$, we can identify the effect ratio and interpret it as the weighted average of the unit causal effect of the exposure on the treatment among individuals whose exposure was affected by the instrument. This is formalized in Proposition \ref{prop:0}.

\begin{proposition} \label{prop:0} Suppose the IV assumptions, (A1)-(A3), and SUTVA in Section 2.3 of the main manuscript holds and the exposure ranges from $0,1,2,\ldots,M$ where $M$ is an integer. Further suppose that the monotonicity assumption where $d_{1ij} \geq d_{0ij}$ holds for all $i,j$. Then, 
\begin{align*}
\lambda =& \frac{\sum_{i=1}^I \sum_{j=1}^{n_i} E(R_{ij} | Z_{ij} = 1, \mathcal{F}, \mathcal{Z}) - E(R_{ij} | Z_{ij} = 0, \mathcal{F}, \mathcal{Z})}{\sum_{i=1}^I \sum_{j=1}^{n_i} E(D_{ij} | Z_{ij} = 1, \mathcal{F}, \mathcal{Z}) - E(D_{ij} | Z_{ij} = 0, \mathcal{F}, \mathcal{Z})} \\
=& \frac{\sum_{i=1}^I \sum_{j=1}^{n_i} \sum_{k=1}^M (r_{ij}^{(k)} - r_{ij}^{(k-1)}) \chi(d_{1ij} \geq k > d_{0ij})}{\sum_{i=1}^I \sum_{j=1}^{n_i} \sum_{k=1}^M \chi(d_{1ij} \geq k > d_{0ij})}
\end{align*}
\end{proposition}
\begin{proof}[Proof of Proposition \ref{prop:0}] By (A3), we have 
\begin{align*}
&E(R_{ij} | Z_{ij} = 1, \mathcal{F}, \mathcal{Z}) - E(R_{ij}| Z_{ij} = 0, \mathcal{F}, \mathcal{Z}) \\ =&r_{1ij}^{(d_{1ij})} - r_{0ij}^{(d_{0ij})} \\
=& \sum_{k=0}^M r_{1ij}^{(k)} \chi(d_{1ij} = k) - \sum_{k=0}^M r_{0ij}^{(k)} \chi(d_{0ij} = k) \\
=& \sum_{k=0}^M r_{1ij}^{(k)} \{\chi(d_{1ij} \geq k) - \chi(d_{1ij} \geq k + 1)\} - \sum_{k=0}^M r_{0ij}^{(k)} \{\chi(d_{0ij} \geq k) - \chi(d_{0ij} \geq k + 1)\} 
\end{align*}
By (A2), $r_{1ij}^{(k)} = r_{0ij}^{(k)}$ for all $k$. Then, we have
\begin{align*}
& \sum_{k=0}^M r_{ij}^{(k)} \{\chi(d_{1ij} \geq k) - \chi(d_{1ij} \geq k+1) - \chi(d_{0ij} \geq k) + \chi(d_{0ij} \geq k+1) \} \\
=& \sum_{k=0}^M r_{ij}^{(k)} \{ \chi(d_{1ij} \geq k) - \chi(d_{0ij} \geq k) \} - \sum_{k=0}^M r_{ij}^{(k)} \{ \chi(d_{1ij} \geq k+1) - \chi(d_{0ij} \geq k+1) \} \\
=& \sum_{k=1}^M r_{ij}^{(k)} \{ \chi(d_{1ij} \geq k) - \chi(d_{0ij} \geq k)\} - \sum_{k=1}^M r_{ij}^{(k-1)} \{\chi(d_{1ij} \geq k) - \chi(d_{0ij} \geq k) \} \\
=& \sum_{k=1}^M (r_{ij}^{(k)} - r_{ij}^{(k-1)}) \{\chi(d_{1ij} \geq k) - \chi(d_{0ij} \geq k)\}
\end{align*}
By monotonicity, $d_{1ij} \geq d_{0ij}$ for all $i,j$. Then,
\begin{align*}
&\sum_{k=1}^M (r_{ij}^{(k)} - r_{ij}^{(k-1)}) \{\chi(d_{1ij} \geq k) - \chi(d_{0ij} \geq k) \} \\
=& \sum_{k=1}^M (r_{ij}^{(k)} - r_{ij}^{(k-1)}) \chi\{ \chi(d_{1ij} \geq k) - \chi(d_{0ij} \geq k) = 1\} \\
=& \sum_{k=1}^M (r_{ij}^{(k)} - r_{ij}^{(k-1)}) \chi(d_{1ij} \geq k > d_{0ij})
\end{align*}
Similarly, by (A3), the expected differences between $Z_{ij} = 1$ and $Z_{ij} = 0$ for the exposure $D_{ij}$ can be written as 
\begin{align*}
&E(D_{ij} | Z_{ij} = 1, \mathcal{F}, \mathcal{Z}) - E(D_{ij} | Z_{ij} = 0, \mathcal{F}, \mathcal{Z}) \\
=& d_{1ij} - d_{0ij} \\
=& \sum_{k=0}^M k \chi(d_{1ij} = k) - \sum_{k=0}^M k \chi(d_{0ij} = k) \\
=& \sum_{k=0}^M k \{\chi(d_{1ij} \geq k) - \chi(d_{1ij} \geq k+1)\} - \sum_{k=0}^M k \{\chi(d_{0ij} \geq k) - \chi(d_{0ij} \geq k+1)\} \\
=& \sum_{k=0}^M k \{\chi(d_{1ij} \geq k) - \chi(d_{1ij} \geq k+1) - \chi(d_{0ij} \geq k) + \chi(d_{0ij} \geq k+1)\} \\
=& \sum_{k=0}^M k \{\chi(d_{1ij} \geq k) - \chi(d_{0ij} \geq k)\} - \sum_{k=0}^M k \{\chi(d_{1ij} \geq k+1) - \chi(d_{0ij} \geq k+1)\} \\
=& \sum_{k=1}^M k \{\chi(d_{1ij} \geq k) - \chi(d_{0ij} \geq k)\} - \sum_{k=1}^M (k -1)\{\chi(d_{1ij} \geq k) - \chi(d_{0ij} \geq k)\} \\
=& \sum_{k=1}^M \{\chi(d_{1ij} \geq k) - \chi(d_{0ij} \geq k)\}
\end{align*}
By monotonicity, we have 
\[
\sum_{k=1}^M \{\chi(d_{1ij} \geq k) - \chi(d_{0ij} \geq k)\} = \sum_{k=1}^M \chi(d_{1ij} \geq k > d_{0ij})
\]
Thus, we end up with 
\begin{align*}
&\frac{\sum_{i=1}^I \sum_{j=1}^{n_i} E(R_{ij} | Z_{ij} = 1,\mathcal{F}, \mathcal{Z}) - E(R_{ij} | Z_{ij} = 0, \mathcal{F}, \mathcal{Z})}{\sum_{i=1}^I \sum_{j=1}^{n_i} E(D_{ij} | Z_{ij} = 1, \mathcal{F}, \mathcal{Z}) - E(D_{ij} | Z_{ij} = 0, \mathcal{F}, \mathcal{Z})} \\
=& \frac{\sum_{i=1}^I \sum_{j=1}^{n_i} r_{1ij}^{(d_{1ij})} - r_{0ij}^{(d_{0ij})}}{ \sum_{i=1}^I \sum_{j=1}^{n_i} d_{1ij} - d_{0ij}} \\
=& \frac{\sum_{i=1}^I \sum_{j=1}^{n_i} \sum_{k=1}^M (r_{ij}^{(k)} - r_{ij}^{(k-1)}) \chi(d_{1ij} \geq k > d_{0ij})}{\sum_{i=1}^I \sum_{j=1}^{n_i} \sum_{k=1}^M \chi(d_{1ij} \geq k > d_{0ij})}
\end{align*}
\end{proof} 

\section{Supplementary Materials: Theoretical Properties of Test Statistic} \label{sec.supp:theoryEffectRatio}
Recall the hypothesis of interest from the main manuscript (Section 2.5)
\[
H_0: \lambda = \lambda_0, \quad{} H_a: \lambda \neq \lambda_0
\]
Proposition \ref{prop:1} provides an asymptotic  distribution for the test statistic $T(\lambda_0)$ in equation \eqref{eq:genTestStat} under the null hypothesis $H_0$.
\begin{proposition} \label{prop:1} Assume that for every $I$, (i) $n_i$ remains bounded and (ii) $\frac{1}{I} \sum_{i=1}^I \sum_{j=1}^{n_i} r_{1ij}^{(d_{1ij})} - r_{0ij}^{(d_{0ij})}$ and $\frac{1}{I} \sum_{i=1}^{I}  \sum_{j=1}^{n_i} d_{1ij} - d_{0ij}$ remains fixed at $\bar{r}$ and $\bar{d} \neq 0$, respectively, so that $\bar{\lambda} = \bar{r} / \bar{d}$. In addition, we assume the following moment conditions
\begin{equation} \label{eq:momentCondGeneral}
\sum_{i=1}^I E\{V_{i}^4(\bar{\lambda}) | \mathcal{F}, \mathcal{Z}\} = o(I^2), \quad{} \limsup_{I \to \infty} \frac{\sum_{i=1}^{I} E| V_{i}(\bar{\lambda}) - \mu_{i,\bar{\lambda}}|^3}{ \left[\sum_{i=1}^I Var\{V_i(\bar{\lambda})\} \right]^{3/2}} = 0
\end{equation}
Then, under the null hypothesis $H_0: \lambda =\bar{\lambda}$, for all $t > 0$,
\[
\limsup_{I \to \infty} P \left\{\frac{T(\bar{\lambda})}{S(\bar{\lambda})} \leq -t | \mathcal{F}, \mathcal{Z} \right\} \leq \Phi(-t), \quad{} \limsup_{I \to \infty} P \left\{\frac{T(\bar{\lambda})}{S(\bar{\lambda})} \geq t | \mathcal{F}, \mathcal{Z} \right\} \leq \Phi(-t)
\]
where $\Phi(\cdot)$ is the standard normal distribution.  
\end{proposition}

To prove Proposition \ref{prop:1}, we require the following two Lemmas. Lemma \ref{lem:moments} characterizes the moments of the test statistics in \eqref{eq:genTestStat}. Lemma \ref{lem:biasVar} derives the bias of $S^2(\lambda_0)$ in estimating the variance of $T(\lambda_0)$. Proof of these Lemmas are in Section \ref{sec:proofLemmas} of the Supplementary Materials.
\begin{lemma} \label{lem:moments}
The expected value and the variance of the test statistic in equation \eqref{eq:genTestStat} are
\begin{align*}
E\{T(\lambda_0) | \mathcal{F},\mathcal{Z} \} &=  \frac{1}{I} (\lambda - \lambda_0) \sum_{i=1}^I \sum_{j=1}^{n_i} (d_{1ij} - d_{0ij}) \\
Var\{T(\lambda_0) | \mathcal{F},\mathcal{Z} \} &= \frac{1}{I^2} \sum_{i=1}^I  \frac{1}{n_i}  \sum_{i=1}^{n_i} (a_{ij,\lambda_0} - \bar{a}_{i,\lambda_0})^2
\end{align*}
where 
\[
a_{ij,\lambda_0} =   \frac{n_i}{m_i} y_{1ij,\lambda_0}^{(d_{1ij})} + \frac{n_i}{n_i - m_i} y_{0ij,\lambda_0}^{(d_{0ij})} , \quad{} \bar{a}_{i,\lambda_0} = \frac{1}{n_i} \sum_{i=1}^{n_i} a_{ij,\lambda_0}
\]
\end{lemma}
\begin{lemma} \label{lem:biasVar}
Let $\mu_{i,\lambda_0} = E\{V_{i}(\lambda_0) | \mathcal{F}, \mathcal{Z}\}$ and $\mu_{\lambda_0} = E\{T(\lambda_0) | \mathcal{F},\mathcal{Z}\}$. The bias of \eqref{eq:genEstVar} in estimating the variance of the test statistic in \eqref{eq:genTestStat} is
\begin{equation} \label{eq:biasVar}
E\{S^2(\lambda_0) | \mathcal{F}, \mathcal{Z}\} - Var\{T(\lambda_0) | \mathcal{F},\mathcal{Z}\} = \frac{1}{I(I-1)} \sum_{i=1}^I ( \mu_{i,\lambda_0} - \mu_{\lambda_0})^2
\end{equation}
\end{lemma}
\begin{proof}[Proof of Proposition \ref{prop:1}] We use the same notation adopted in the proof of Lemma \ref{lem:biasVar}, mainly $\mu_{i,\bar{\lambda}}$, $\mu_{\bar{\lambda}}$, and $v_{i,\bar{\lambda}}^2$. In addition, let $q_{i,\bar{\lambda}} = E\{V_{i}^2(\bar{\lambda}) | \mathcal{F}, \mathcal{Z}\}$, and $v_{\bar{\lambda}} = Var\{T(\bar{\lambda}) | \mathcal{F}, \mathcal{Z}\}$. First, $\sum_{i=1}^I V_{i}^2(\bar{\lambda})/I$ is an unbiased estimator for $\sum_{i=1}^I q_{i,\bar{\lambda}} /I$. In addition, 
\[
Var\left\{\frac{1}{I} \sum_{i=1}^I V_{i}^2(\bar{\lambda}) | \mathcal{F}, \mathcal{Z} \right\} \leq \frac{1}{I^2} \sum_{i=1}^I E\{V_i^4(\bar{\lambda}) | \mathcal{F}, \mathcal{Z}\}
\]
By the fourth moment condition in \eqref{eq:momentCondGeneral}, we have $ \sum_{i=1}^I  V_i^2(\bar{\lambda})/I - \sum_{i=1}^{I}  q_{i,\bar{\lambda}} /I \to 0$ in probability. Similarly, the same fourth moment condition in \eqref{eq:momentCondGeneral} and the same reasoning gives $T(\bar{\lambda}) - \mu_{\bar{\lambda}} \to 0$ in probability because of the growth of the variance of $T(\bar{\lambda})$ is controlled by the moment condition. Since $\mu_{\bar{\lambda}} = 0$ for all $I$ under the null hypothesis, we have, by the continuous mapping theorem,  $T^2(\bar{\lambda}) \to 0$ in probability. Combining all these convergence results, we get that for $\epsilon > 0$ and $\delta > 0$, there exists $I^*$ such that
\[ \text{for $I \geq I^*$:}
P\left\{\frac{1}{I} \sum_{i=1}^I V_{i}^2(\bar{\lambda}) - \frac{1}{I} \sum_{i=1}^I q_{i,\bar{\lambda}} < -\frac{\epsilon}{2} \right\} < \frac{\delta}{2}, \quad{} P\left\{T^2(\bar{\lambda}) < -\frac{\epsilon}{2} \right\} < \frac{\delta}{2}
\]
and 
\begin{align*}
&P\left\{ I S^2(\bar{\lambda}) -  I v_{\bar{\lambda}} < -\epsilon  \right\} \\
=& P\left[ \frac{I}{I-1} \left\{\frac{1}{I} \sum_{i=1}^I  V_{i}^2(\bar{\lambda}) - T^2(\bar{\lambda}) \right\} - I v_{\bar{\lambda}}  < - \epsilon \right] \\
=& P\left[ \frac{I}{I-1} \left\{\frac{1}{I} \sum_{i=1}^I V_{i}^2(\bar{\lambda}) - \frac{1}{I} \sum_{i=1}^I  q_{i,\bar{\lambda}} + \frac{1}{I} \sum_{i=1}^I  q_{i,\bar{\lambda}} - T^2(\bar{\lambda}) \right\} - Iv_{\bar{\lambda}} < - \epsilon \right] \\
=& P\left[ \frac{I}{I-1} \left\{\frac{1}{I} \sum_{i=1}^I  V_{i}^2(\bar{\lambda}) - \frac{1}{I} \sum_{i=1}^I  q_{i,\bar{\lambda}} - T^2(\bar{\lambda}) \right\} - Iv_{\bar{\lambda}} + \frac{1}{I-1} \sum_{i=1}^I  q_{i,\bar{\lambda}} < - \epsilon \right] \\
\leq& P\left[ \frac{I}{I-1} \left\{\frac{1}{I} \sum_{i=1}^I V_{i}^2(\bar{\lambda}) - \frac{1}{I} \sum_{i=1}^I q_{i,\bar{\lambda}} - T^2(\bar{\lambda}) \right\}  < - \epsilon \right] \\
\leq& \frac{\delta}{2} + \frac{\delta}{2} = \delta
\end{align*}
Stated in words, $IS^2(\bar{\lambda})$ will over-estimate $Iv_{\bar{\lambda}}$ with high probability. 

Second, under the null hypothesis $H_0: \lambda = \bar{\lambda}$ and from Lemma \ref{lem:moments}, $ \sum_{i=1}^I \mu_{i,\bar{\lambda}} /I = 0$. Hence, we can rewrite the test statistic as 
\[
T(\bar{\lambda}) = \frac{1}{I} \sum_{i=1}^I V_{i}(\bar{\lambda}) = \frac{1}{I} \sum_{i=1}^I  [V_{i}(\bar{\lambda}) - \mu_{i,\bar{\lambda}}]
\] 
where the test statistic becomes a sum of independent random variables $V_{i}(\bar{\lambda}) - \mu_{i,\bar{\lambda}}$ with mean zero and variance $v_{i,\bar{\lambda}}$.

Finally, combining the two facts, under the null $H_0: \lambda =\bar{\lambda}$, we have
\[
\frac{T(\bar{\lambda})}{S(\bar{\lambda})} = \left[ \frac{ \frac{1}{I}\sum_{i=1}^I \{V_{i}(\bar{\lambda}) - \mu_{i,\bar{\lambda}}\}}{\sqrt{ \frac{1}{I^2} \sum_{i=1}^I  v_{i,\bar{\lambda}}}} \right] \left\{ \frac{\sqrt{ \frac{1}{I^2} \sum_{i=1}^I  v_{i,\bar{\lambda}}}}{\sqrt{S^2(\bar{\lambda})}} \right\}
\]
By conditions specified in \citet[pg 186]{breiman_probability_1992} for the central limit theorem with non-identical distributions, the first parenthesis term converges to the standard Normal distribution. From our result about $IS_{\bar{\lambda}}^2$ overestimating $Iv_{\bar{\lambda}}$, the second parenthesis term will be smaller than 1 with high probability. Hence, taking the $\sup$ of the entire expression, we obtain
\[
\limsup_{I \to \infty} P \left\{\frac{T(\bar{\lambda})}{S(\bar{\lambda})} \leq -t | \mathcal{F}, \mathcal{Z} \right\} \leq \Phi(-t), \quad{} \limsup_{I \to \infty} P \left\{\frac{T(\bar{\lambda})}{S(\bar{\lambda})} \geq t | \mathcal{F}, \mathcal{Z} \right\} \leq \Phi(-t)
\]
where $\Phi()$ is the standard normal distribution.  
\end{proof}

Proposition \ref{prop:1} provides a way to estimate the effect ratio, compute p-values, and calculate confidence intervals. In particular, in the spirit of \citet{hodges_estimation_1963}, the estimator for the effect ratio, denoted as $\hat{\lambda}$, is the solution to the equation $T(\hat{\lambda})/ S(\hat{\lambda}) = 0$. The 95\% confidence interval for the effect ratio is the solution to the equation $T(\lambda) / S(\lambda) = \pm 1.96$. Corollary \ref{coro:effectRatioEstimatorFormula} presents a solution to the equation $T(\hat{\lambda}) / S(\hat{\lambda}) = q$ for any value of $q$.
\begin{corollary} \label{coro:effectRatioEstimatorFormula} For any value $q$, the solution to $T(\lambda) / S(\lambda) = q$ is a solution to the quadratic equation $A_2 \lambda^2 + A_1 \lambda + A_0 = 0$ where 
\begin{align*}
A_2 &= \bar{H}_{.}^2 - \frac{q^2}{I(I-1)} \sum_{i=1}^I (H_i - \bar{H}_{.})^2 \\
A_1 &= -2\bar{G}_{.} \bar{H}_{.} + \frac{2q^2}{I(I-1)} \left\{ \sum_{i=1}^I (G_i - \bar{G}_{.})(H_i - \bar{H}_{.}) \right\}  \\
A_0 &= \bar{G}_{.}^2 - \frac{q^2}{I(I-1)} \sum_{i=1}^I (G_i - \bar{G}_{.})^2 
\end{align*}
where
\begin{align*}
G_i &=  \frac{n_i^2}{m_i (n_i - m_i)} \sum_{j=1}^{n_i} (Z_{ij} - \bar{Z}_{i.})(R_{ij} - \bar{R}_{i.}) \\
H_i &= \frac{n_i^2}{m_i (n_i - m_i)} \sum_{j=1}^{n_i} (Z_{ij} - \bar{Z}_{i.})(D_{ij} - \bar{D}_{i.}) \\
\bar{Z}_{i.} &= \frac{1}{n_i} \sum_{j=1}^{n_i} Z_{ij} , \quad{} \bar{D}_{i.} = \frac{1}{n_i} \sum_{j=1}^{n_i} D_{ij}, \quad{} \bar{R}_{i.} = \frac{1}{n_i} \sum_{j=1}^{n_i} R_{ij} \\ \bar{H}_{.} &= \frac{1}{I} \sum_{i=1}^I H_i, \quad{} \bar{G}_{.} = \frac{1}{I} \sum_{i=1}^I G_i
\end{align*}
\end{corollary}
\begin{proof}
First, we see that $T(\lambda)/S(\lambda) = q$ implies $T^2(\lambda) = q^2 S^2(\lambda)$. This expression can be rewritten as
\begin{equation} \label{eq:effectRatioConfIntFormula1}
T^2(\lambda) = \frac{q^2}{I(I-1)}\sum_{i=1}^I (V_i(\lambda) - T(\lambda))^2 = \frac{q^2}{I(I-1)} \left \{\sum_{i=1}^I  V_i^2(\lambda) - I T^2(\lambda)\right\}
\end{equation}
Rearranging the terms in \eqref{eq:effectRatioConfIntFormula1}, we get
\[
T^2(\lambda) \left(1 + \frac{q^2}{I-1}\right) = \frac{q^2}{I(I-1)} \sum_{i=1}^I V_i^2(\lambda)
\]
Second, we can re-express $V_{i}(\lambda)$ as follows.
\begin{align*}
V_{i}(\lambda) &= \sum_{j=1}^{n_i} \left(\frac{n_i}{m_i} + \frac{n_i}{n_i - m_i}\right)  Z_{ij} R_{ij} - \sum_{j=1}^{n_i} \frac{n_i}{n_i - m_i} R_{ij} \\
&\quad{} - \sum_{j=1}^{n_i} \left(\frac{n_i}{m_i} + \frac{n_i}{n_i - m_i}\right)\lambda Z_{ij} D_{ij} + \sum_{j=1}^{n_i} \frac{n_i}{n_i - m_i} \lambda D_{ij} \\
 &= \sum_{j=1}^{n_i} \frac{n_i^2}{m_i(n_i - m_i)}  Z_{ij} R_{ij} - \left(\sum_{j=1}^{n_i} \frac{n_i}{n_i - m_i} R_{ij}\right) \left(\frac{1}{m_i} \sum_{j=1}^{n_i} Z_{ij} \right) \\
&\quad{} - \sum_{j=1}^{n_i} \frac{n_i^2}{m_i(n_i - m_i)} \lambda Z_{ij} D_{ij} + \left(\sum_{j=1}^{n_i} \frac{n_i}{n_i - m_i} \lambda D_{ij} \right)\left( \frac{1}{m_i} \sum_{j=1}^{n_i} Z_{ij}\right) \\
 &= \frac{n_i^2}{m_i (n_i - m_i)} \left(\sum_{j=1}^{n_i} Z_{ij} R_{ij} - \frac{1}{n_i} \sum_{j=1}^{n_i} R_{ij}\sum_{j=1}^{n_i} Z_{ij}\right) \\
&\quad{} - \lambda \frac{n_i^2}{m_i (n_i - m_i)} \left( \sum_{j=1}^{n_i} Z_{ij} D_{ij} - \frac{1}{n_i} \sum_{j=1}^{n_i} D_{ij} \sum_{j=1}^{n_i} Z_{ij} \right) \\
 &= \frac{n_i^2}{m_i (n_i - m_i)} \sum_{j=1}^{n_i} (Z_{ij} - \bar{Z}_{i.})(R_{ij} - \bar{R}_{i.}) - \lambda \frac{n_i^2}{m_i (n_i - m_i)} \sum_{j=1}^{n_i} (Z_{ij} - \bar{Z}_{i.})(D_{ij} - \bar{D}_{i.})
\end{align*}
Immediately, we also have $ V_i(\lambda) = G_i - \lambda H_i$. Then, we can rewrite $\sum_{i=1}^I V_i^2(\lambda) $ and $T^2(\lambda)$ as follows
\begin{align*}
\sum_{i=1}^I  V_i^2(\lambda) &= \sum_{i=1}^I (G_i - \lambda H_i)^2 \\
&= \sum_{i=1}^I G_i^2  - 2 \lambda \sum_{i=1}^I G_i H_i + \lambda^2 \sum_{i=1}^I H_i^2 \\
T^2(\lambda) &= \frac{1}{I^2} \left\{ \sum_{i=1}^I  V_i(\lambda) \right\}^2 \\
&= \frac{1}{I^2} \left\{ \sum_{i=1}^I (G_i - \lambda H_i) \right\}^2 \\
&= \frac{1}{I^2} \left\{ \left(\sum_{i=1}^I G_i\right)^2 - 2\lambda \sum_{i=1}^I G_i \sum_{i=1}^I H_i + \lambda^2 \left(\sum_{i=1}^I H_i\right)^2 \right\}
\end{align*}
Overall, we can rewrite the equation \eqref{eq:effectRatioConfIntFormula1} as 
\begin{align*}
&\frac{1}{I^2} \left\{ \left(\sum_{i=1}^I G_i\right)^2 - 2\lambda \sum_{i=1}^I G_i \sum_{i=1}^I H_i + \lambda^2 \left(\sum_{i=1}^I H_i\right)^2 \right\}\left(1 + \frac{q^2}{I-1}\right) \\
=& \frac{q^2}{I(I-1)} \left(\sum_{i=1}^I G_i^2  - 2 \lambda \sum_{i=1}^I G_i H_i + \lambda^2 \sum_{i=1}^I H_i^2\right)
\end{align*}
Finally, we pull out the coefficients associated with $\lambda^2$ and $\lambda$, denoted as $A_2$ and $A_1$, respectively. The remaining term are constants and we denote them as $A_0$. All $A_2$, $A_1$, and $A_0$ are explicitly written below.
\begin{align*}
A_2 &= \frac{1}{I^2} \left( \sum_{i=1}^I H_i\right)^2 + \frac{q^2}{I(I-1)} \left\{ \frac{1}{I} \left(\sum_{i=1}^I H_i \right)^2 - \sum_{i=1}^I H_i^2 \right\} \\
 &= \bar{H}_{.}^2 - \frac{q^2}{I(I-1)} \sum_{i=1}^I (H_i - \bar{H}_{.})^2 \\
A_1 &= -2 \left[ \frac{1}{I^2} \sum_{i=1}^I G_i \sum_{i=1}^{I} H_i + \frac{q^2}{I(I-1)} \left\{\frac{1}{I} \sum_{i=1}^I G_i \sum_{i=1}^I H_i - \sum_{i=1}^I G_i H_i \right\}\right] \\
&= -2 \left[ \bar{G}_{.} \bar{H}_{.} - \frac{q^2}{I(I-1)} \left\{ \sum_{i=1}^I (G_i - \bar{G}_{.})(H_i - \bar{H}_{.}) \right\} \right] \\
A_0 &= \frac{1}{I^2} \left( \sum_{i=1}^I G_i\right)^2 + \frac{q^2}{I(I-1)} \left\{ \frac{1}{I} \left(\sum_{i=1}^I G_i \right)^2 - \sum_{i=1}^I G_i^2 \right\} \\
 &= \bar{G}_{.}^2 - \frac{q^2}{I(I-1)} \sum_{i=1}^I (G_i - \bar{G}_{.})^2 
\end{align*}
\end{proof}
If $q = 0$ in Corollary \ref{coro:effectRatioEstimatorFormula}, there is only one solution to the quadratic equation since 
\[
A_2 \lambda^2 + A_1 \lambda + A_0 = \bar{H}_{.}^2 \lambda^2 - 2 \bar{H}_{.} \bar{G}_{.} \lambda + \bar{G}_{.}^2 = (\bar{H}_{.}\lambda - \bar{G}_{.})^2 = 0
\]
This gives us an explicit formula for the estimator of the effect ratio, denoted as $\hat{\lambda}$.
\begin{equation} \label{eq:effectRatioEstimate}
\hat{\lambda} = \frac{\bar{G}_{.}}{\bar{H}{.}} = \frac{\sum_{i=1}^I \frac{n_i^2}{m_i (n_i - m_i)} \sum_{j=1}^{n_i} (R_{ij} - \bar{R}_{i.})(Z_{ij} - \bar{Z}_{i.})}{\sum_{i=1}^I  \frac{n_i^2}{m_i(n_i - m_i)} \sum_{j=1}^{n_i} (D_{ij} - \bar{D}_{i.})(Z_{ij} - \bar{Z}_{i.})}
\end{equation}

\section{Supplementary Materials: Sensitivity Analysis} \label{sec:sens}

To model this deviation from randomized assignment due to unmeasured confounders, let $\pi_{ij} = P(Z_{ij} = 1 | \mathcal{F})$ and $\pi_{ik} = P(Z_{ik} = 1 | \mathcal{F})$ for each unit $j$ and $k$ in the $i$th matched set. The odds that unit $j$ will receive $Z_{ij} =1$ instead of $Z_{ij} = 0$ is $\pi_{ij}/(1 - \pi_{ij})$. Similarly, the odds for unit $k$ is $\pi_{ik}/(1 - \pi_{ik})$. Suppose the ratio of these odds is bounded by $\Gamma \geq 1$
\begin{equation} \label{eq:gamma}
\frac{1}{\Gamma} \leq \frac{\pi_{ij}  (1 -\pi_{ik})}{\pi_{ik} (1 - \pi_{ij})} \leq \Gamma
\end{equation}
If unmeasured confounders play no role in the assignment of $Z_{ij}$, $\Gamma = 1$ and $\pi_{ij} = \pi_{ik}$. That is, child $j$ and $k$ have the same probability of receiving $Z_{ij} =1$ in matched set $i$. If there are unmeasured confounders that affect the distribution of $Z_{ij}$, then $\pi_{ij} \neq \pi_{ik}$ and $\Gamma > 1$. By \citet{rosenbaum_observational_2002}, equation \eqref{eq:gamma} is equivalent to 
\begin{equation} \label{eq:z}
P(\mathbf{Z} = \mathbf{z} | \mathcal{F},\mathcal{Z}) = \frac{\exp(\gamma \mathbf{z}^T \mathbf{u})}{\sum_{b \in \Omega} \exp(\gamma \mathbf{b}^T \mathbf{u})} 
\end{equation}
where $\mathbf{u} = (u_{11},u_{12},...,u_{In_{I}})$. Unfortunately, the exact probability of \eqref{eq:z} is unknown as it depends on the vector of unobserved confounders, $(u_{11},\ldots,u_{In_I})$. However, for a fixed $\Gamma > 1$, we can obtain lower and upper bounds on \eqref{eq:z}. Furthermore, since the inference on the effect ratio $\lambda$ is derived from the distribution of $P(\mathbf{Z} = \mathbf{z} | \mathcal{F}, \mathcal{Z})$, these bounds can be used to compute a range of possible p-values under the null hypothesis. The range of p-values indicates the effect of unmeasured confounders on the conclusions reached by the inference on $\lambda$. If the range contains $\alpha$, the significance value, then we cannot reject the null hypothesis at the $\alpha$ level when there is an unmeasured confounder with an effect quantified by $\Gamma$.

Specifically, consider Fisher's sharp null hypothesis, $H_0: r_{1ij}^{(d_{1ij})} = r_{0ij}^{(d_{0ij})}$ for all $i=1,\ldots,n$ and $j=1,\ldots,n_i$. Note that this hypothesis implies the hypothesis $H_0: \lambda = 0$. Furthermore, the test statistic in \eqref{eq:genTestStat} simplifies to
\begin{align*}
T(0) &= \frac{1}{I} \sum_{i=1}^I \left\{\frac{n_i}{m_i}  \sum_{j=1}^{n_i} Z_{ij} R_{ij} - \frac{n_i}{n_i -m_i} \sum_{j=1}^{n_i} (1 - Z_{ij}) R_{ij} \right\} \\
&= \frac{1}{I} \sum_{i=1}^I \frac{n_i^2}{m_i(n_i - m_i)} \sum_{j=1}^{n_i} Z_{ij} R_{ij} - \frac{1}{I}\sum_{i=1}^{I} \frac{n_i}{n_i - m_i} \sum_{j=1}^{n_i} R_{ij}
\end{align*}
Regardless of the distribution of $P(\mathbf{Z} = \mathbf{z} | \mathcal{F}, \mathcal{Z})$, $\frac{1}{I} \sum_{i=1}^I n_i / (n_i - m_i) \sum_{j=1}^{n_i} R_{ij}$ is a constant since $r_{1ij}^{(d_{1ij})} = r_{0ij}^{(d_{0ij})}$ under Fisher's sharp null hypothesis. Hence, we can use the simpler statistic, $\tilde{T}(0)$,
\begin{equation} \label{eq:signscore}
\tilde{T}(0) = \frac{1}{I}\sum_{i=1}^I \frac{n_i}{m_i (n_i - m_i)} \sum_{j=1}^{n_i} Z_{ij} R_{ij}  
\end{equation}
to test the Fisher's sharp null hypothesis. If the responses are binary, equation \eqref{eq:signscore} is the sign-score test statistic for which exact bounds on p-values exist \citep{rosenbaum_observational_2002}. If the responses are continuous, \citet{gastwirth_asymptotic_2000} and \citet{small_simultaneous_2009} provide an approximate bound on p-values.

\section{Supplementary Materials: Amplification of Sensitivity Analysis}
Following \citet{rosenbaum_amplification_2009}, we can also reinterpret the sensitivity parameter $\Gamma$ by considering a binary unmeasured confounder with two values $\Delta$ and $\Lambda$ where $\Delta$ and $\Lambda$ have the following property
\begin{equation} \label{eq:amplifySens}
\Gamma = \frac{\Delta \Lambda + 1}{\Delta + \Lambda}, \quad{} \Delta > 0, \Lambda > 0
\end{equation}
The parameter $\Lambda$ refers to the odds of having one instrument value over another. The parameter $\Delta$ refers to the odds of having one outcome over another. For each $\Gamma$, we can use equation \eqref{eq:amplifySens} and translate the interpretation of $\Gamma$ as the combined effect an unmeasured confounder must have on the instrument, $\Lambda$, and on the outcome, $\Delta$, to change the inference.

\begin{figure}
\vspace{6pc}
\includegraphics[scale = 0.6]{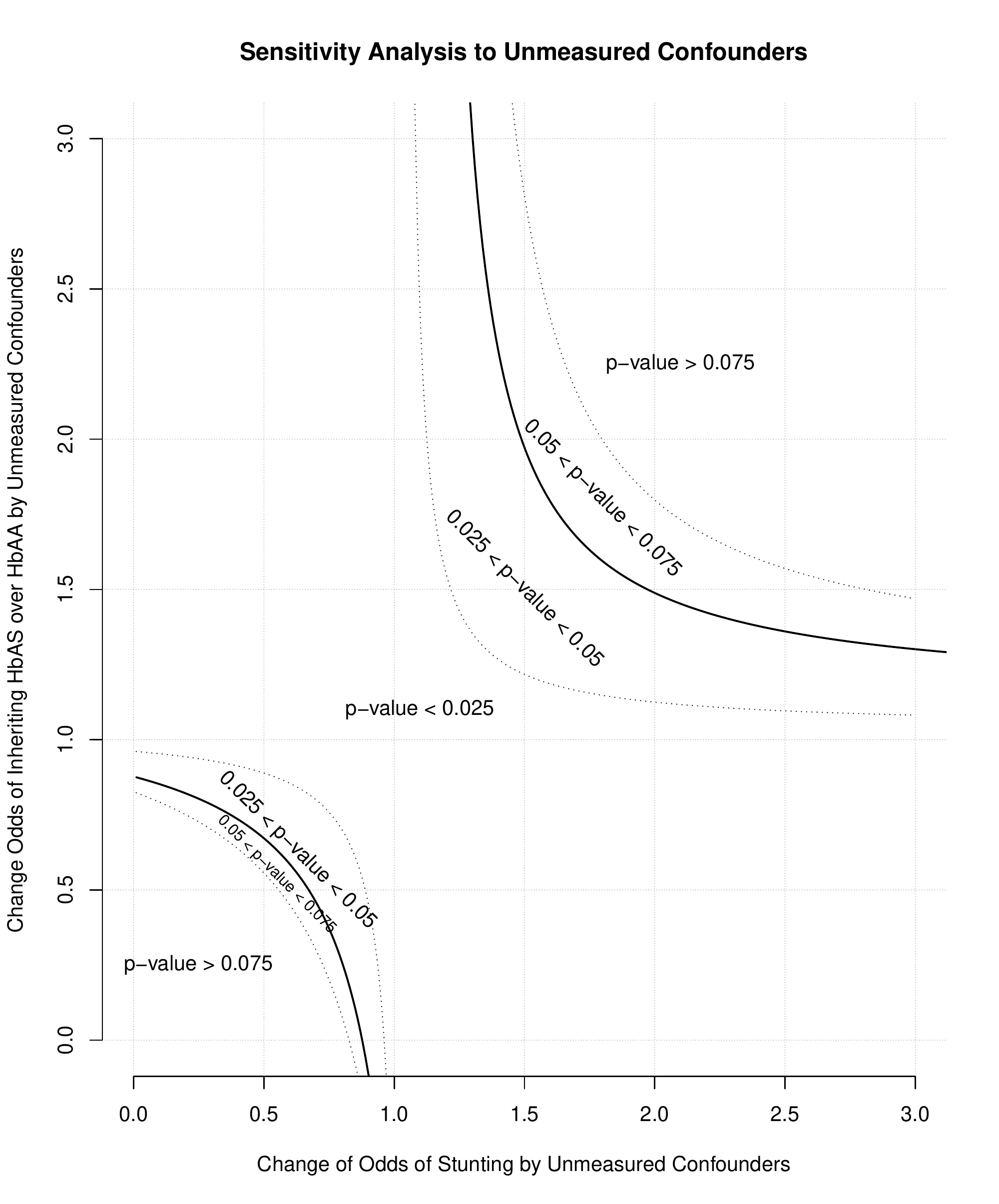}
\caption[]{Amplification of sensitivity analysis. Each point on the graph represents an effect by an unmeasured confounder on the instrument (HbAS) and on the outcome (stunting) to change the inference, specifically the p-value. Points within the two bold curves correspond to effects by unmeasured confounders that will give us p-values $<$ 0.05 and points outside the two bold curves correspond to effects that will give us p-values $>$ 0.05, thereby retaining our null hypothesis.}
\label{fig:amplifiedSens}
\end{figure}

Figure \ref{fig:amplifiedSens} shows the result of applying the amplification of $\Gamma$ by looking at the effect by unmeasured confounders on the odds of stunting and odds of inheriting HbAS over HbAA and on the inference. Specifically, the different values of $\Gamma$ in the sensitivity analysis provides us with range of possible p-values. By equation \eqref{eq:amplifySens}, each $\Gamma$ is associated with two other sensitivity parameters $\Delta$, odds of stunting, and $\Lambda$, odds of inheriting HbAS over HbAA, and can be presented as a two-dimensional plot with each axis representing $\Delta$ and $\Lambda$. For example, the point $(\Delta = 1.5, \Lambda = 1.5)$ on Figure \ref{fig:amplifiedSens} represents an unmeasured confounder that increases the odds of stunting and inheriting HbAS over HbAA by a factor of 1.5 and produces a p-value in between 0.025 and 0.05, which does not contain the significance level of 0.05. Hence, the null hypothesis would still be rejected despite having such an unmeasured confounder. In contrast, if the unmeasured confounder had an effect of $(2.0,2.0)$ specified on the plot, the null hypothesis would be retained since the p-value contains the significance level of 0.05.

\section{Supplementary Materials: Efficiency} \label{sec:efficiency}
\subsection{Formula for efficiency}
One of the advantages of full matching is its flexibility to accommodate various sizes of matched sets. All things being equal in terms of covariate balance, we would like an estimator of the effect ratio $\lambda$ that is as efficient as possible. This is particularly the case with full matching where an unconstrained full matching can create large matched sets which reduces efficiency \citep{hansen_full_2004}. However, we can constrain full matching to increase efficiency by restricting matched sets to have a maximum number of controls and/or treated units per matched set \citep{hansen_full_2004}. This section studies the statistical efficiency of the estimator for $\lambda$ in equation \eqref{eq:effectRatioEstimate} under different constraints on full matching.

To study the efficiency of the effect ratio estimator for different $n_i$ and $m_i$, we study a simple version of the structural equations model introduced popular in econometrics and has been used to study the properties of 2SLS, the most popular IV estimator \citep{wooldridge_econometrics_2010}. Let $(R_{ij},D_{ij}, Z_{ij})$ be i.i.d. observations from an infinite population under the following model.
\begin{align} \label{eq:eff1}
R_{ij} &= \alpha_i + \beta D_{ij} + \epsilon_{ij}, \quad{} E(\epsilon_{ij} | Z_{ij} ) = 0\\
D_{ij} &= \tau_i + \gamma Z_{ij} + \xi_{ij}, \quad{} E(\xi_{ij} | Z_{ij}) = 0 \label{eq:eff2}
\end{align}
with the following moment conditions.
\begin{align*}
Var(\epsilon_{ij} | Z_{ij} ) = \sigma_{i,R}^2, \quad{} Var(\xi_{ij} | Z_{ij} ) = \sigma_{i,D}^2, \quad{} E(\epsilon_{ij} \xi_{ij} | Z_{ij}) = \sigma_{i,RD}
\end{align*}
The parameters $\alpha_i, i=1,\ldots,I$ measure the effect on the outcome from being in matched set $i$. The parameter $\beta$ is the effect of interest, the effect of the exposure on the outcome. Note that the treatment effect in \eqref{eq:eff1} is assumed to be homogeneous for everyone, which is not assumed in the main manuscript. The parameters $\tau_i, i=1,\ldots,I$ measure the effect on the exposure from being in matched set $i$. The parameter $\gamma$ is the effect of the instrument on the exposure. By including $\alpha_i$ and $\tau_i$, the models \eqref{eq:eff1} and \eqref{eq:eff2} incorporate the matching aspect of IV estimation since each matched set $i$ have effects on $R_{ij}$ and $D_{ij}$ that are unique to that matched set. 

The effect ratio, $\lambda$, is related to parameters found in standard structural equation models in \eqref{eq:eff1} and \eqref{eq:eff2}. To illustrate this, note that the potential outcomes notation can be rewritten under the models \eqref{eq:eff1} and \eqref{eq:eff2} as follows.
\begin{align*}
R_{ij} &= \begin{cases} r_{1ij}^{(d_{1ij})} = 
 \alpha_i + \beta \tau_i + \beta \gamma + \beta \xi_{ij} + \epsilon_{ij} &\text{if $Z_{ij} = 1$}\\
r_{0ij}^{(d_{0ij})} = \alpha_i + \beta\tau_i + \beta \xi_{ij} + \epsilon_{ij} & \text{if $Z_{ij} = 0$}
\end{cases} \\
D_{ij} &= 
\begin{cases}
d_{1ij} = \tau_i + \gamma + \xi_{ij} & \text{if $Z_{ij} = 1$} \\
d_{0ij} = \tau_i + \xi_{ij} & \text{if $Z_{ij} = 0$}
\end{cases}
\end{align*}
Then, the effect ratio in \eqref{eq:effectRatioGeneral} turns out to be
\begin{align*}
\lambda =  \frac{ \sum_{i=1}^I \sum_{j=1}^{n_i} r_{1ij}^{(d_{1ij})} - r_{0ij}^{(d_{0ij})}}{\sum_{i=1}^I \sum_{j=1}^{n_i} d_{1ij} - d_{0ij}} = \frac{ \sum_{i=1}^I \sum_{j=1}^{n_i} \beta \gamma}{\sum_{i=1}^I  \sum_{j=1}^{n_i} \gamma} = \frac{\beta \gamma}{\gamma} = \beta
\end{align*}
Hence, $\lambda = \beta$ and because of this equivalence, inferences for the effect ratio provides inference for $\beta$.

Thus, the parameter $\beta$ can be estimated by the effect ratio estimator discussed in Section \ref{sec.supp:theoryEffectRatio} of the Supplementary Materials, specifically equation \eqref{eq:effectRatioEstimate},
\[
\hat{\beta} = \frac{\sum_{i=1}^I \frac{n_i^2}{m_i (n_i - m_i)} \sum_{j=1}^{n_i} (Z_{ij} - \bar{Z}_{i.}) (R_{ij} - \bar{R}_{i.})}{\sum_{i=1}^I \frac{n_i^2}{m_i(n_i - m_i)} \sum_{j=1}^{n_i} (Z_{ij} - \bar{Z}_{i.}) (D_{ij} - \bar{D}_{i.})}
\]
Proposition \ref{prop:eff} computes the asymptotic variance of $\hat{\beta}$ to study the efficiency of the effect ratio estimator.
\begin{proposition} \label{prop:eff} Suppose we have models \eqref{eq:eff1} and \eqref{eq:eff2} with $\gamma \neq 0$ and the third moment of $\epsilon_{ij}$ is bounded for all $i,j$. Define the following variables
\begin{align*}
J_i &= \sum_{j=1}^{n_i} (Z_{ij} - \bar{Z}_{i.})(\epsilon_{ij} - \bar{\epsilon}_{i.}), \quad{} H_i = \sum_{j=1}^{n_i} (Z_{ij} - \bar{Z}_{i.})(D_{ij} - \bar{D}_{i.}), \quad{}
\bar{\epsilon}_{i.} = \frac{1}{n_i} \sum_{j=1}^{n_i} \epsilon_{ij} \\
s_I^2 &= \sum_{i=1}^I \frac{n_i^3}{m_i (n_i - m_i)} \sigma_{i,R}^2
\end{align*}
Assume that (i) $Z_{ij}$ are fixed, (ii) $n_i$ remain bounded for all $i$, and the following moment conditions are met for $J_i$ and $H_i$
\[
\limsup_{I \to \infty} \frac{1}{s_I^{3}} \sum_{i=1}^I \frac{n_i^6}{m_i^3(n_i - m_i)^3} E(|J_i|^3) = 0, \quad{}
\sum_{i=1}^I Var \left( \frac{n_i^2}{m_i (n_i - m_i)} H_i^2 \right) = o(I^2)
\]
Then, the asymptotic variance of the effect ratio estimator in \eqref{eq:effectRatioEstimate} is 
\[
\sqrt{I}(\hat{\beta} - \beta) \to N\left\{0, \frac{ \left(\lim_{I \to \infty} \frac{s_I}{\sqrt{I}} \right)^2}{\gamma^2 \left(\lim_{I \to \infty} \frac{1}{I}  \sum_{i=1}^I n_i \right)^2} \right\}
\]
\end{proposition}
\begin{proof}[Proof of Proposition \ref{prop:eff}] 
First, for all $i=1,\ldots,I$ and $j = 1,\ldots,n_i$, we have
\[
Z_{ij} - \bar{Z}_{i.} = \begin{cases}
1 - \frac{m_i}{n_i} & \text{if $Z_{ij} = 1$} \\
- \frac{m_i}{n_i} & \text{if $Z_{ij} = 0$}
\end{cases}
\] 
Furthermore,
\[
\sum_{j=1}^{n_i} (Z_{ij} - \bar{Z}_{i.}) = 0, \quad{} \sum_{j=1}^{n_i} (Z_{ij} - \bar{Z}_{i.})^2 = \frac{m_i (n_i - m_i)}{n_i} 
\]
Second, for fixed $Z_{ij}$, we have the following expected values for $J_i$ 
\begin{align*}
E(J_i) =& 0 \\
E(J_i^2) =& Var \left\{\sum_{j=1}^{n_i} (Z_{ij} - \bar{Z}_{i.})(\epsilon_{ij} - \bar{\epsilon}_{i.}) \right\} \\
=& \sum_{j=1}^{n_i} (Z_{ij} - \bar{Z}_{i.})^2 Var(\epsilon_{ij} - \bar{\epsilon}_{i.})  + \sum_{j,k} (Z_{ij} - \bar{Z}_{i.}) (Z_{ik} - \bar{Z}_{k.}) Cov(\epsilon_{ij} - \bar{\epsilon}_{i.},\epsilon_{ik} - \bar{\epsilon}_{i.}) \\
=& (1 - \frac{1}{n_i}) \sigma_{i,R}^2 \sum_{j=1}^{n_i} (Z_{ij} - \bar{Z}_{i.})^2 - \frac{1}{n_i} \sigma_{i,R}^2 \sum_{j,k} (Z_{ij} - \bar{Z}_{i.}) (Z_{ik} - \bar{Z}_{k.}) \\
=& \sigma_{i,R}^2 \sum_{j=1}^{n_i} (Z_{ij} - \bar{Z}_{i.})^2 - \frac{1}{n_i} \sigma_{i,R}^2 \left\{\sum_{j=1}^{n_i}(Z_{ij} - \bar{Z}_{i.})\right\}^2 \\
=& \sigma_{i,R}^2 \frac{m_i (n_i - m_i)}{n_i} 
\end{align*}
For the third moment, for each $i$, let $k_1,\ldots,k_{n_i}$ be non-negative integers and define the multinomial coefficient as follows.
\[
\binom{3}{k_1,\ldots,k_{n_i}} = \frac{3!}{k_1! \cdots k_{n_i}!}
\]
Then, we have
\begin{align*}
E(|J_i^3|) =&E|\left\{ \sum_{j=1}^{n_i} (Z_{ij} - \bar{Z}_{i.})(\epsilon_{ij} - \bar{\epsilon}_{i.}) \right\}^3| \\
=& E| \sum_{k_1 + \cdots + k_{n_i} = 3} \binom{3}{k_1,\ldots,k_{n_i}} \prod_{j=1}^{n_i} \left\{(Z_{ij} - \bar{Z}_{i.})(\epsilon_{ij} - \bar{\epsilon}_{i.})\right\}^{k_j} | \\
\leq& \sum_{k_1 + \cdots + k_{n_i} = 3} \binom{3}{k_1,\ldots,k_{n_i}} \prod_{j=1}^{n_i} |Z_{ij} - \bar{Z}_{i.}|^{k_j} E|\epsilon_{ij} - \bar{\epsilon}_{i.}|^{k_j} < \infty
\end{align*}
because third moments exist and are bounded for all $\epsilon_{ij}$ and $n_i$ is bounded. Third, based on these moment calculations, it immediately follows that
\[
E \left[ \sum_{i=1}^I \left\{\frac{ n_i^2}{(m_i)(n_i - m_i)} J_i \right\}^2 \right] = \sum_{i=1}^I \left\{\frac{ n_i^4}{(m_i)^2(n_i - m_i)^2} \right\} \left\{\frac{m_i(n_i - m_i)}{n_i} \sigma_{i,R}^2 \right\} = s_I^2 
\]
Then, by Theorem 9.2 in Chapter 9, Section 3 of \citet{breiman_probability_1992} (pg 187), the sum of $J_i$ weighted by $n_i^2 / m_i (n_i - m_i)$ is a standard Normal distribution  
\[
\frac{\sum_{i=1}^I \frac{ n_i^2}{m_i (n_i -m_i)} J_i}{s_I} \to N(0,1)
\]
Fourth, for $H_i$, we have the following moments 
\begin{align*}
E(H_i) =& \gamma  m_i (1 - \frac{m_i}{n_i}) \\
Var(H_i) =& Var \left(\sum_{j=1}^{n_i} (Z_{ij} - \bar{Z}_{i.})(D_{ij} - \bar{D}_{i.}) \right) \\
=& (1 - \frac{1}{n_i}) \sigma_{i,D}^2 \sum_{j=1}^{n_i} (Z_{ij} - \bar{Z}_{i.})^2 - \frac{1}{n_i} \sigma_{i,D}^2 \sum_{j,k} (Z_{ij} - \bar{Z}_{i.}) (Z_{ik} - \bar{Z}_{k.}) \\
=& \sigma_{i,D}^2 \sum_{j=1}^{n_i} (Z_{ij} - \bar{Z}_{i.})^2 - \frac{1}{n_i} \sigma_{i,D}^2 \left(\sum_{j=1}^{n_i}(Z_{ij} - \bar{Z}_{i.})\right)^2 \\
=& \sigma_{i,D}^2 \frac{m_i (n_i - m_i)}{n_i}
\end{align*}
Fifth, by Theorem C in page 27 of \citet{serfling_approximation_1980}, 
\begin{align*}
&\frac{1}{I} \sum_{i=1}^I \frac{n_i^2}{m_i(n_i - m_i)} H_i - \gamma \frac{1}{I} \sum_{i=1}^I E\left\{ \frac{n_i^2}{m_i (n_i - m_i)} H_i \right\} \\
=& \frac{1}{I} \sum_{i=1}^I \frac{n_i^2}{m_i(n_i - m_i)} H_i - \gamma \frac{1}{I} \sum_{i=1}^I n_i \to 0  
\end{align*}
Finally, combining all these facts together, we can rewrite the effect ratio estimator as follows.
\begin{align*}
\hat{\beta} &= \frac{\sum_{i=1}^I \frac{n_i^2}{m_i (n_i - m_i)} \sum_{j=1}^{n_i} (Z_{ij} - \bar{Z}_{i.}) (R_{ij} - \bar{R}_{i.})}{\sum_{i=1}^I \frac{n_i^2}{m_i(n_i - m_i)} \sum_{j=1}^{n_i} (Z_{ij} - \bar{Z}_{i.}) (D_{ij} - \bar{D}_{i.})} \\
&=\beta + \frac{\sum_{i=1}^I \frac{n_i^2}{m_i (n_i - m_i)} \sum_{j=1}^{n_i} (Z_{ij} - \bar{Z}_{i.}) (\epsilon_{ij} - \bar{\epsilon}_{i.})}{\sum_{i=1}^I \frac{n_i^2}{m_i(n_i - m_i)} H_i} \\
&= \beta + \frac{\sum_{i=1}^I \frac{n_i^2}{m_i(n_i - m_i)} J_i}{\sum_{i=1}^{I} \frac{n_i^2}{m_i(n_i - m_i)} H_i}
\end{align*}
which leads to 
\[
\sqrt{I} (\hat{\beta} - \beta) = \left\{\frac{\sum_{i=1}^I \frac{n_i^2}{m_i(n_i - m_i)} J_i}{s_I} \right\} \left\{\frac{\frac{1}{\sqrt{I}} s_I}{\frac{1}{I}\sum_{i=1}^{I} \frac{n_i^2}{m_i(n_i - m_i)} H_i} \right\}
\]
Finally, using Slutsky's Theorem, $\sqrt{I}(\hat{\beta} - \beta)$ converges to a Normal distribution with mean $0$ and stated asymptotic variance.
\end{proof}
Proposition \ref{prop:eff} provides an easy way to compare between different types of full matching methods and their effect on the estimation of the effect ratio. For example, in the simple case of homoscedastic variance, the approximate variance of $\hat{\lambda}$ is
\[
Var(\hat{\lambda}) \approx K \frac{\sum_{i=1}^{I} \frac{n_i^3}{n_i -1}}{\left(\sum_{i=1}^I n_i \right)^2}
\]
where $K$ is some constant that depends on the variance of $R_{ij}$ and the strength of the instrument. Since $K$ will be identical for all full matched designs, we can simply look at the quantities to the right of $K$ to tweak our full matching algorithm to produce the most efficient estimator.

With regards to the quality of the approximation, the asymptotic variance is a decent approximation to the estimator's variance if the number of matched sets, $I$, are large or if the instruments are strong. This is demonstrated in Table \ref{tab:efficiencyRatio} which is a result of the following simulation study. The variables $R_{ij}, D_{ij}$ and $Z_{ij}$ are generated via the model in \eqref{eq:eff1} and \eqref{eq:eff2} with $Z_{ij}$ assumed to be fixed. We randomly pick $\alpha_i, \tau_i$, and $\beta$. We pick $\gamma$ to be $1$ for the strong instrument case and $-0.2$ for the weak instrument case. We assume a homoscedastic variance for the error terms where all the $\sigma_{i,R}^2, \sigma_{i,D}^2$, and $\sigma_{i,RD}$ are the same for every $i$. We compute the effect ratio estimator, repeat this process $1000$ times, and compute the simulated variance. The theoretical variance is calculated based on the formula provided in Proposition \ref{prop:eff}.

\begin{table}
\caption{ Comparison of simulated variance and theoretical variance for different strength of instruments and matched set number $I$.}
\label{tab:efficiencyRatio}
\begin{tabular}{l l l l l}
$I$ & \multicolumn{2}{c}{Theoretical Variance} & \multicolumn{2}{c}{Simulated Variance}  \\  \cline{2-3}  \cline{4-5} 
 &  Strong & Weak & Strong & Weak \\ \hline
50 & 0.024 & 0.59 & 0.028 & 3224.30\\
100 & 0.012 & 0.30 & 0.012 & 181.06\\
110 & 0.011 & 0.27 & 0.012 & 2506.92\\
500 & 0.0024 & 0.060 & 0.0025 & 2.05\\
1000 & 0.0012 & 0.030 & 0.0012 & 0.037\\
5000 & 0.00024 & 0.0060 & 0.00024 & 0.0063\\
10000 & 0.00012 & 0.0030& 0.00012 & 0.0030\\
\end{tabular}
\end{table}

Table \ref{tab:efficiencyRatio} shows us that for strong instruments, the agreement between theoretical formula in Proposition \ref{prop:eff} and simulation is quite good for all values of $I$. On the other hand, for weak instruments, there is substantial deviation between the theoretical variance and the simulated variance until $I$ is above $5000$.

\subsection{Simulation to approximate efficiency}
The prior section offers a formula to compute efficiency of various full matching schemes. However, for the formula to be valid, it requires, among other things, a linear model between the outcome, $R_{ij}$ and the exposure $D_{ij}$. In our study where stunting, the outcome, is a binary variable and malaria, the exposure, is a whole number, it is unreasonable to assume that $R_{ij}$ is a linear function of $D_{ij}$.

In such cases, we propose a simulation study to analyze efficiency for different full matching schemes. As an illustration, consider our study with the effect of malaria on stunting. For each matching scheme, we fix $Z_{ij}$ and $X_{ij}$, which, in turn, fixes the matched sets. For the other variables, $D_{ij}$ and $R_{ij}$, we assume a Poisson relationship between $D_{ij}$ and $Z_{ij}$ and a logistic relationship between $D_{ij}$ and $R_{ij}$. In particular, we use the following model 
\[
P(R_{ij} = 1) = \frac{1}{1 + e^{-(\alpha_i + \beta D_{ij} + u_{ij})}}, \quad{} E(D_{ij}) = e^{\tau_i + \gamma Z_{ij}}
\]
We fix $\beta$, the effect of malaria on stunting, to be $0.32$ and $\gamma$, the strength of the instrument, to be $-0.20$ based on the estimates in \citet{kang_causal_2013}; the estimate of $\gamma$ was based on the risk ratio estimate. We also randomly choose $\alpha_i$ and $\tau_i$, the intercepts, from Normal distributions with means $-1.67$ and $-0.19$, respectively, and variances $0.12$ and $0.027$, respectively. The mean and the variance for $\alpha_i$ is from the intercept term and its corresponding standard error of the logistic regression between $R_{ij}$ and $D_{ij}$. Similarly, the mean and the variance for $\tau_i$ is from the intercept term and its corresponding standard error of the Poisson regression between $D_{ij}$ and $Z_{ij}$. Once all the parameters are set, we sample $884$ observations of $(R_{ij},D_{ij})$ (i.e. the sample size of the malaria data set) and compute the effect ratio estimator based on the sample of 884. Note that the effect ratio estimator should be able to estimate $\beta$ since it doesn't rely on the functional form between stunting (i.e. outcome) and malaria episodes (i.e. exposure). We repeat the simulation $5000$ times and compute the median absolute deviation as a robust proxy for variance of the effect ratio estimator. 

\begin{table}
\caption{Trade-off between efficiency and balance for different full matching schemes that use all the data based on simulation based on median absolute deviation and standardized bias.}
\label{tab:effData}
\begin{tabular}{l l l}
Matching & Median absolute deviation & Standardized bias \\ \hline
Full matching (max strata size is 9) & 0.90 & 0.23 \\
Full matching (max strata size is 10) & 0.96 & 0.19 \\
Full matching (max strata size is 15) & 0.97 & 0.10 \\
Full matching (unrestricted) & 0.98 & 0.055 \\
\end{tabular}
\end{table}

Table \ref{tab:effData} shows the trade-off between efficiency and covariate balance for different full matching schemes that use all $884$ samples of the malaria data. In particular, we restrict the matched set sizes to different values to see their impact on efficiency and standardized bias. The standardized bias is the instrumental propensity score \citep{cheng_using_2011} and is calculated as the difference in propensity scores before and after matching normalized by the within group standard deviation before matching (the square root of the average of the variances within the group).
We see that unrestricted full matching has the lowest bias among all other full matching schemes. However, full matching with restricted strata size of 9 has the lowest median absolute deviation, albeit by a little in comparison to other matching schemes. Given the large bias reduction by using unrestricted full matching with a small gain in median absolute deviation, we use unrestricted full matching in our main manuscript.
 
\section{Supplementary Materials: Extended Simulation}
\subsection{Strength of instruments}
The simulation setup is identical to the one in Section 3 of the main manuscript. We present another aspect of our estimator's performance in relation to 2SLS, specifically the median absolute deviation (MAD). Figure \ref{fig:strengthMAD} measures the MAD of 2SLS and our method. Our method tends to have a slightly higher MAD than 2SLS. This higher variability of our method is to be expected since our method uses a nonparametric approach whereas 2SLS is a parametric approach. However, as the instrument gets stronger (i.e. high concentration parameter), the gap between the two MADs shrinks quickly. 

\begin{figure}
\vspace{6pc}
\includegraphics[scale=0.55]{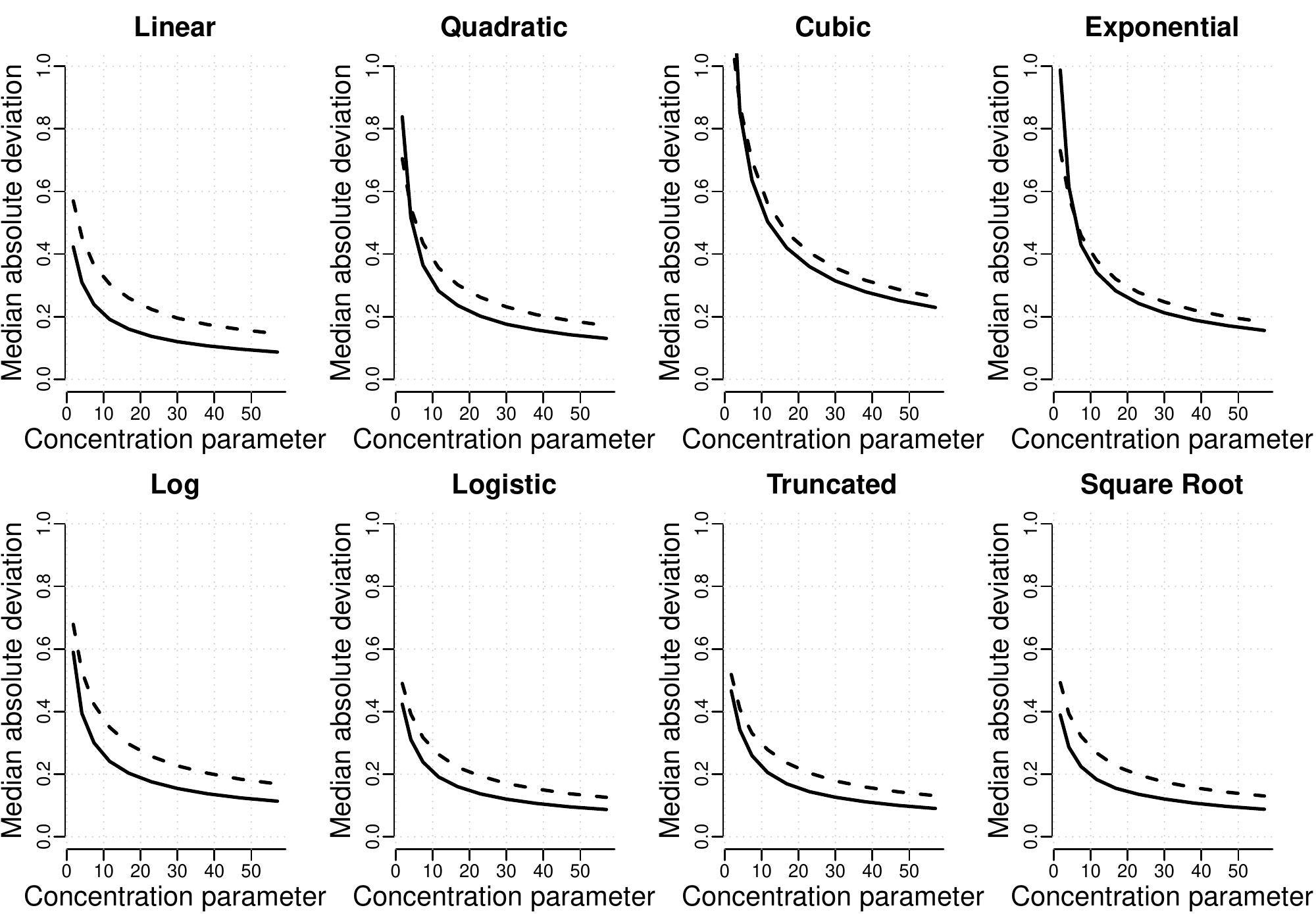}
\caption[]{Median absolute deviation between our method and two stage least squares (2SLS) for different concentration parameters. The solid line indicates 2SLS and the dotted line indicates our method.}
\label{fig:strengthMAD}
\end{figure}

\subsection{Sample size}
The simulation result presented here has the identical setup as the one in Section 3 of the main manuscript. However, we fix the strength of the instrument to be very strong, but vary the sample size. We keep the ratio between $Z_{ij} = 1$ to $Z_{ij} = 0$ to be $1$ to $7$, respectively. We compare the performance of 2SLS and our method with respect to bias, variance, and type I error rate as we vary $f(\cdot)$. 

\begin{figure}
\vspace{6pc}
\includegraphics[scale=0.6]{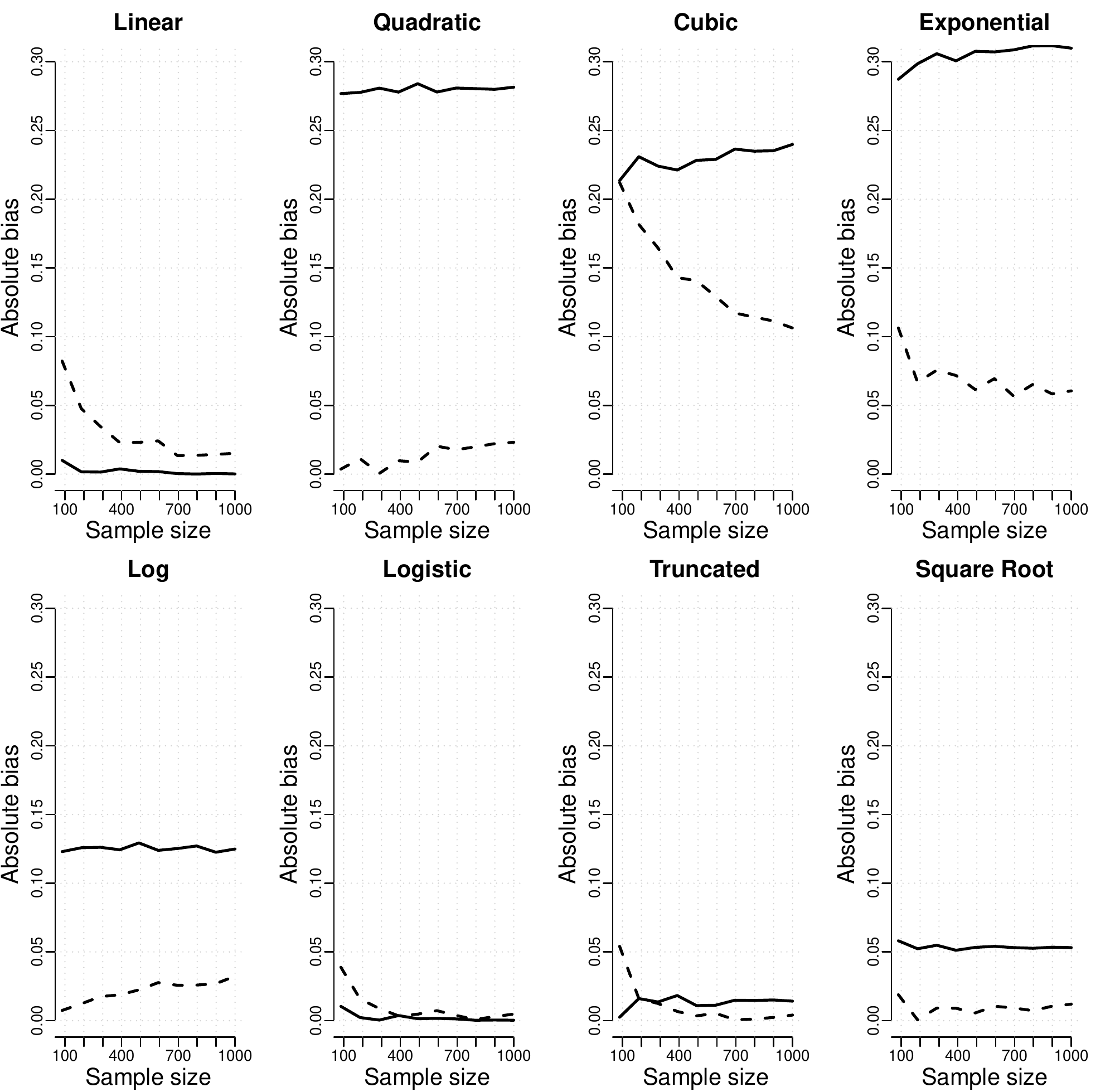}
\caption{Absolute bias of the median between our method and two stage least squares (2SLS) for different sample sizes. The solid line indicates 2SLS and the dotted line indicates our method.}
\label{fig:samplesizeBias}
\end{figure}

Figure \ref{fig:samplesizeBias} measures the absolute bias of 2SLS and our method. When $f(\cdot)$ is a linear function of the observed covariates $\mathbf{x}_{ij}$, 2SLS does better than our method, which is to be expected since 2SLS works best when the model is linear. However, if $f(\cdot)$ is non-linear, our matching estimator does better than 2SLS and is never substantially worse. For example, for quadratic, cubic, exponential, log, and square root functions, our method has lower bias than 2SLS for all sample size. For logistic and truncated functions, our method is similar in performance to 2SLS.

\begin{figure}
\includegraphics[scale=0.6]{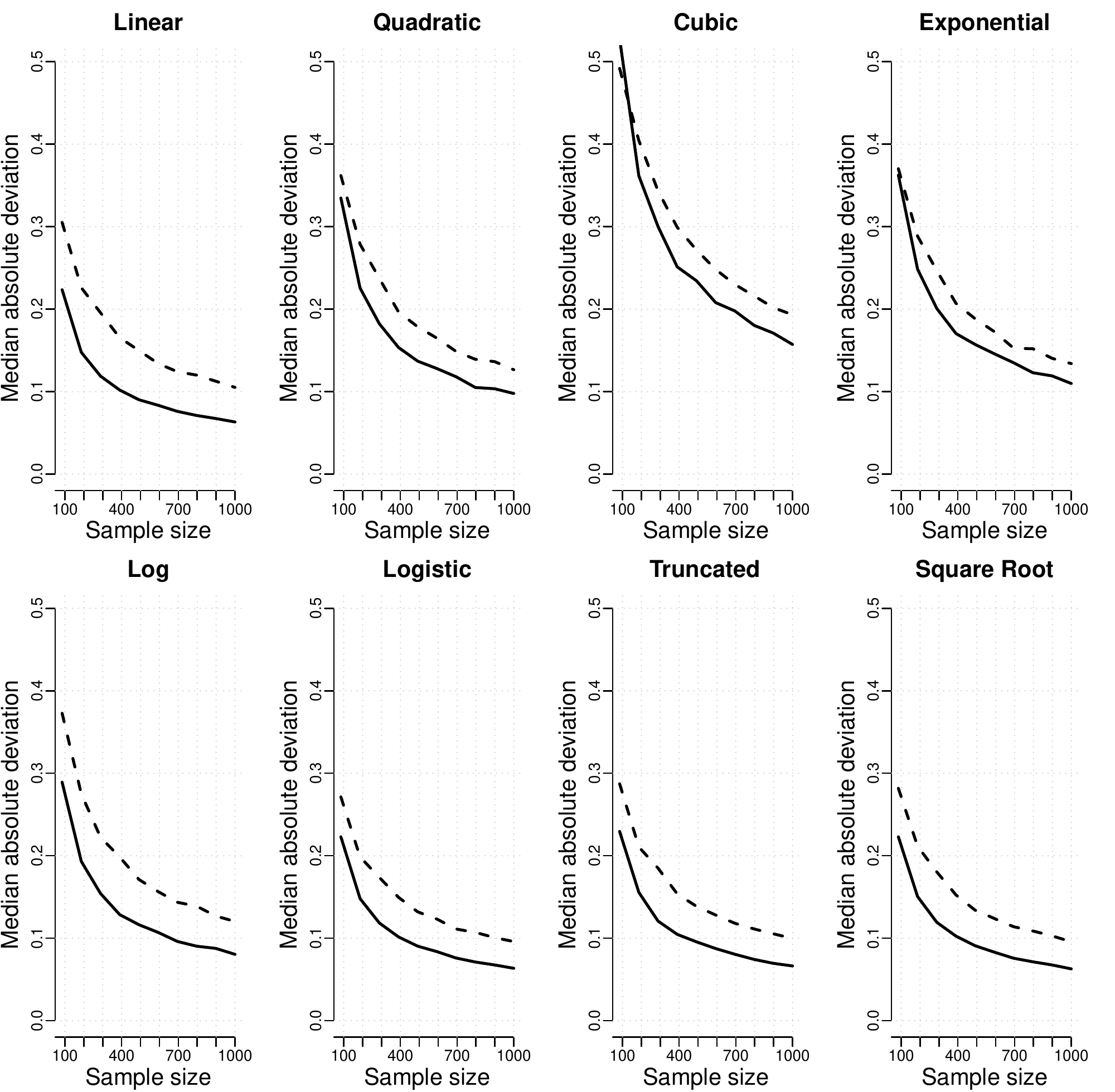}
\caption{Median absolute deviation between our method and two stage least squares (2SLS) for different sample sizes. The solid line indicates 2SLS and the dotted line indicates our method.}
\label{fig:samplesizeMAD}
\end{figure}

Figure \ref{fig:samplesizeMAD} measures the median absolute deviation (MAD) of 2SLS and our method. Our method tends to have a slightly higher MAD than 2SLS. This higher variability of our method is to be expected since our method uses a nonparametric approach whereas 2SLS is a parametric approach. 

\begin{figure}
\includegraphics[scale=0.6]{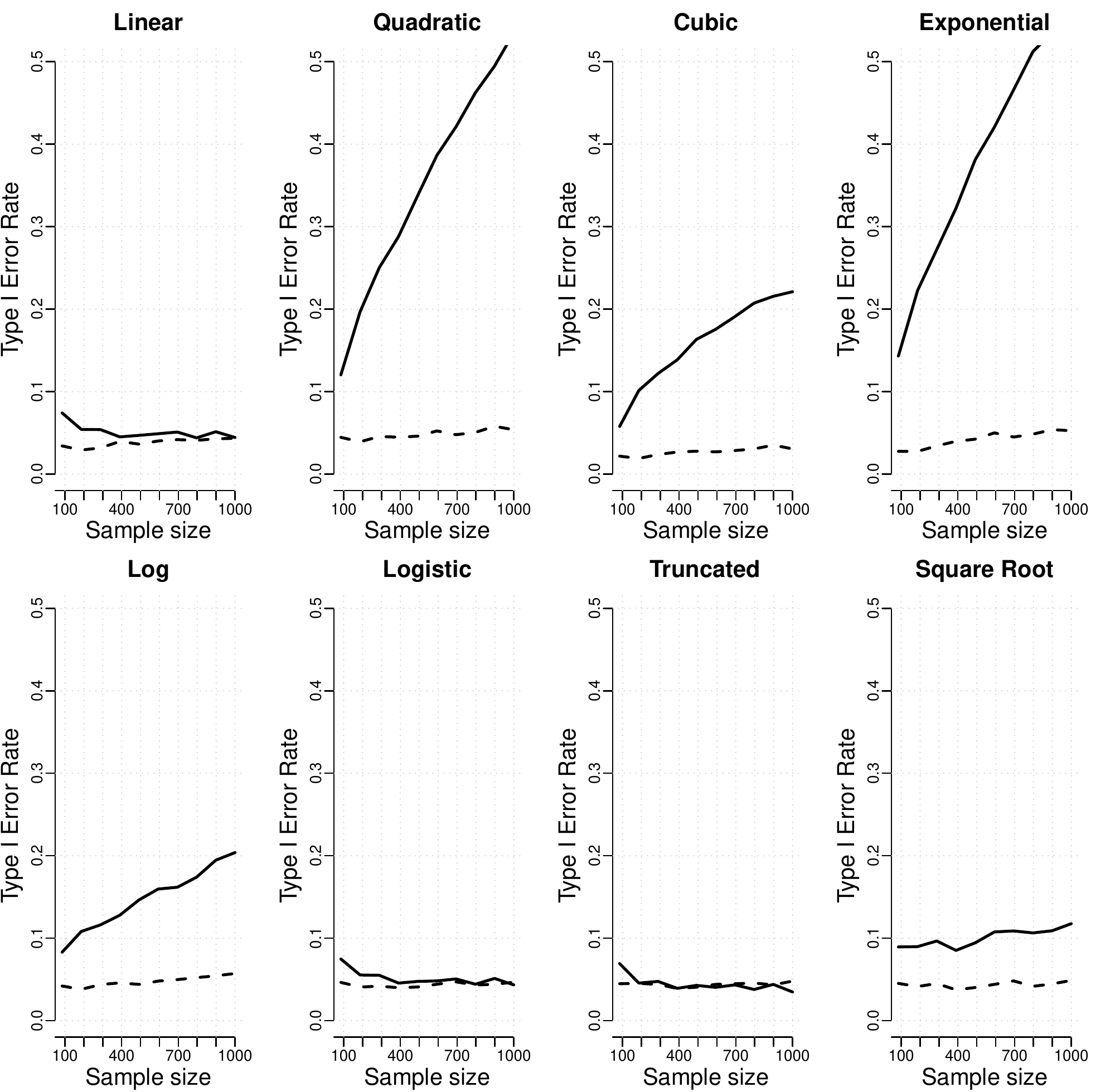}
\caption{ Type I error rate between our method and two stage least squares (2SLS) for different sample sizes. The solid line indicates 2SLS and the dotted line indicates our method.}
\label{fig:samplesizeTypeI}
\end{figure}

Finally, Figure \ref{fig:samplesizeTypeI} measures the Type I error rate of 2SLS and our method. Regardless of the function type and and sample size, our method retains the nominal $0.05$ rate. In fact, even for the linear case where 2SLS is designed to excel, our estimator has the correct Type I error rate for all sample size while 2SLS has higher Type I error for small sample size. For all the non-linear functions, the Type I error rate for 2SLS remains above the 0.05 line, with the notable exception of  logistic and truncated functions whose 2SLS estimators has similar Type I error as our method. In contrast, our estimator maintains the nominal Type I error rate for all sample sizes. This provides evidence that our estimator will have the correct 95\% coverage for confidence intervals regardless of the non-linearity or for different sample size. 

\subsection{Comparison to \citet{frolich_nonparametric_2007}}
In this section, we provide a few additional details of our matching estimator to another non-parametric IV estimator with covariates explored by \citet{frolich_nonparametric_2007} as presented in the main manuscript. The simulation setup is designed to mimic the data type in the malaria data where we have a binary $Z_{ij}$, discrete $D_{ij}$ taking on values $0,1$ and $2$, and a continuous response $R_{ij}$ and is identical to the one presented in the main manuscript. Specifically, we have
\begin{align*}
R_{ij} &=  \alpha + \beta D_{ij} + f(\mathbf{X}_{ij}) + U_{ij} + \epsilon_{ij} \\
D_{ij} &= \chi(D_{ij}^* < -1) + 2\chi(-1 \leq D_{ij}^*  < 1) + 3 \chi(1 \leq D_{ij}^*) \\
D_{ij}^* &= \kappa + \pi Z_{ij} + \bm{\rho}^T \mathbf{X}_{ij} + U_{ij} + \xi_{ij}
\end{align*}
where $D_{ij}^*$ is a latent variable, $U_{ij}$ serve as the unmeasured confounder, and $\epsilon_{ij}, \xi_{ij}, U_{ij}$ are all i.i.d Normal. Covariates $X_{ij}$ are generated similar to the simulation study in the main manuscript. As before, we look at bias and variance across different strengths of instruments, different sample size, and different functions $f(\cdot)$. Similar to the main manuscript, we use the default settings provided in \citet{frolich_estimation_2010}, which implements the method by \citet{frolich_nonparametric_2007}. 

We were not able to produce Type I error results for the method of \citet{frolich_nonparametric_2007} because of a coding error in the code provided by \citet{frolich_estimation_2010} which provided negative standard errors on the estimates produced by it. Fr\"{o}lich (personal communication) is aware of the issue and will be releasing a new version in the future.

\section{Supplementary Materials: Proof to Lemmas} \label{sec:proofLemmas}
\begin{proof}[Proof of Lemma 1] 
Let $y_{0ij,\lambda_0} = r_{0ij}^{(d_{0ij})} - \lambda_0 d_{0ij}$ and $y_{1ij,\lambda_0} = r_{1ij}^{(d_{1ij})} - \lambda_0 d_{1ij}$. Then, $V_i(\lambda_0)$ becomes
\begin{align*}
V_{i}(\lambda_0) &= \frac{n_i}{m_i} \sum_{j=1}^{n_i} Z_{ij} (R_{ij} - \lambda_0 D_{ij}) - \frac{n_i}{n_i - m_i} \sum_{j=1}^{n_i} (1 - Z_{ij}) (R_{ij} - \lambda_0 D_{ij}) \\
&= \frac{n_i}{m_i} \sum_{j=1}^{n_i} Z_{ij} y_{1ij,\lambda_0} - \frac{n_i}{n_i - m_i} \sum_{j=1}^{n_i} (1 - Z_{ij}) y_{0ij,\lambda_0} 
\end{align*}
By assumption (A3) of IV in the main manuscript, $Z_{ij}$ are independent within each strata. Then, for any $i = 1,\ldots,I$ and for $j, k = 1,...,n_i$ where $j \neq k$
\[
E(Z_{ij} | \mathcal{F}, \mathcal{Z}) =\frac{m_i}{n_i}, \quad{} E(Z_{ij} Z_{ik} | \mathcal{F}, \mathcal{Z}) = \frac{m_i (m_i -1)}{n_i (n_i -1)} = \frac{m_i - 1}{n_i} 
\]
where the second equality is true because in full matching, $m_i = 1$ and $n_i = m_i -1$ or $m_i = n_i -1$ and $n_i = 1$.  Then, the expectation of $V_{i}(\lambda_0)$ and the test statistic $T(\lambda_0)$ are
\begin{align*}
E\{V_{i}(\lambda_0)| \mathcal{F}, \mathcal{Z}\} &= \sum_{j=1}^{n_i} (r_{1ij}^{(d_{1ij})} - r_{0ij}^{(d_{0ij})}) - \lambda_0 (d_{1ij} - d_{0ij})  \\
E\{T(\lambda_0) | \mathcal{F},\mathcal{Z} \} &= \frac{1}{I} \sum_{i=1}^I  E\{V_{i}(\lambda_0) | \mathcal{F}, \mathcal{Z} \}= \frac{1}{I} (\lambda - \lambda_0) \sum_{i=1}^I  \sum_{j=1}^{n_i} (d_{1ij} - d_{0ij})
\end{align*}
For variance of $V_{i}(\lambda_0)$, Proposition 2 in \citet[Sec. 2.4.4]{rosenbaum_observational_2002} gives us
\begin{align*}
&Var\{V_{i}(\lambda_0) | \mathcal{F}, \mathcal{Z}\} \\
=& Var \left\{ \sum_{j=1}^{n_i} Z_{ij} \left( \frac{n_i}{m_i} y_{1ij,\lambda_0} + \frac{n_i}{n_i - m_i} y_{0ij,\lambda_0} \right) | \mathcal{F}, \mathcal{Z} \right\}  \\
=& \sum_{j=1}^{n_i} \left(\frac{m_{i}}{n_i} - \frac{m_i^2}{n_i^2} \right) a_{ij,\lambda_0}^2 + \left( \frac{m_i - 1}{n_i} - \frac{m_i^2}{n_i^2} \right) \sum_{j\neq k} a_{ij,\lambda_0} a_{ik,\lambda_0}  \\
=& \left(\frac{m_{i}}{n_i} - \frac{m_i^2}{n_i^2} - \frac{m_i -1}{n_i} + \frac{m_{i}^2}{n_i^2} \right)\sum_{j=1}^{n_i} a_{ij,\lambda_0}^2 + \left( \frac{m_i - 1}{n_i} - \frac{m_i^2}{n_i^2} \right) \sum_{j, k} a_{ij,\lambda_0} a_{ik,\lambda_0} \\
=& \frac{1}{n_i} \sum_{j=1}^{n_i} a_{ij,\lambda_0}^2 + \frac{n_i(m_i - 1) - m_i^2}{n_i^2} \sum_{j,k} a_{ij,\lambda_0} a_{ik,\lambda_0} \\
=& \frac{1}{n_i} \sum_{j=1}^{n_i} a_{ij,\lambda_0}^2 - \frac{1}{n_i^2} \sum_{j,k} a_{ij,\lambda_0} a_{ik,\lambda_0} \\
=& \frac{1}{n_i} \sum_{i=1}^{n_i} (a_{ij,\lambda_0} - \bar{a}_{i,\lambda})^2
\end{align*}
Finally, the variance of $T(\lambda_0)$ is given by
\[
Var\{T(\lambda_0) | \mathcal{F}, \mathcal{Z}\} = \frac{1}{D^2} \sum_{i=1}^I Var \{ V_{i}(\lambda_0) | \mathcal{F},\mathcal{Z} \} = \frac{1}{I^2} \sum_{i=1}^I \frac{1}{n_i}\sum_{j=1}^{n_i} (a_{ij,\lambda_0} - \bar{a}_{i,\lambda})^2
\]
\end{proof}

\begin{proof}[Proof of Lemma 2] Let $v_{i,\lambda_0}^2 = Var \{V_{i}(\lambda_0) | \mathcal{F}, \mathcal{Z}\}$. Under the generalized effect ratio, the bias of the estimator \eqref{eq:genEstVar} is
\begin{align*}
&E\{S^2(\lambda_0) | \mathcal{F}, \mathcal{Z} \} \\
=&\frac{1}{I(I-1)} \sum_{i=1}^I E[\{V_{i}(\lambda_0) - T(\lambda_0)\}^2 | \mathcal{F}, \mathcal{Z}] \\
=&\frac{1}{I(I-1)} \sum_{i=1}^I  E\{V_{i}^2(\lambda_0) | \mathcal{F}, \mathcal{Z}\} + E\{T^2(\lambda_0) | \mathcal{F}, \mathcal{Z}\} - 2E\{V_{i}(\lambda_0) T(\lambda_0) | \mathcal{F}, \mathcal{Z}\} \\
=&\frac{1}{I(I-1)} \sum_{i=1}^I (\mu_{i,\lambda_0}^2 + v_{i,\lambda_0}) + \left(\mu_{\lambda_0}^2 + \frac{1}{I^2} \sum_{j=1}^I v_{j,\lambda_0}\right) \\
&-\frac{2}{I}\left( \mu_{i,\lambda_0}^2 + v_{i,\lambda_0} + \sum_{j \neq i} \mu_{i,\lambda_0} \mu_{j,\lambda_0} \right)  \\
=& \frac{1}{I(I-1)} \sum_{i=1}^I \left( v_{i,\lambda_0} - \frac{2}{I}  v_{i,\lambda_0} + \frac{1}{I^2} \sum_{j=1}^I v_{j,\lambda_0}\right) \\
&+ \frac{1}{I(I-1)} \sum_{i=1}^I \left(  \mu_{i,\lambda_0}^2 + \mu_{\lambda_0}^2 - \frac{2}{I} \sum_{j=1}^I \mu_{i,\lambda_0} \mu_{j,\lambda_0}\right) \\
=& \left(\frac{I^2 - 2I + I}{I(I-1)} \right) \frac{1}{I^2} \sum_{i=1}^n  v_{i,\lambda_0} + \frac{1}{I(I-1)} \sum_{i=1}^I (\mu_{i,\lambda_0} - \mu_{\lambda_0})^2 \\
=& \frac{1}{I^2} \sum_{i=1}^I  v_{i,\lambda_0} + \frac{1}{I(I-1)} \sum_{i=1}^I  (\mu_{i,\lambda_0}- \mu_{\lambda_0})^2
\end{align*}
\end{proof}

\end{document}